\documentclass[journal,onecolumn,peerreview,11pt]{ieeecolor}
\usepackage{generic}
\usepackage{cite}
\usepackage{amsmath,amssymb,amsfonts}
\usepackage{mathtools}
\usepackage{mathrsfs}
\usepackage{algorithmic}
\usepackage{graphicx}
\usepackage{algorithm,algorithmic}
\usepackage{hyperref}
\makeatletter
\let\MYcaption\@makecaption
\makeatother
\usepackage{textcomp}
\usepackage{braket}
\usepackage[euler]{textgreek}
\newtheorem{theorem}{\bf Theorem}
\newtheorem{corollary}{\bf Corollary}
\newtheorem{definition}{\bf Definition}

\newtheorem{proposition}{\bf Proposition}
\newtheorem{assumption}{\bf Assumption}
\newtheorem{lemma}{\bf Lemma}
\newtheorem{remark}{\bf Remark}
\let\labelindent\relax
\usepackage{enumitem}
\usepackage{tikz}
\usepackage{overpic}
\usepackage[dvipsnames]{xcolor}

\newcommand{\J}{\mathscr{J}}
\newcommand{\R}{\mathbb{R}}
\newcommand{\NN}{\mathbb{N}}
\newcommand{\E}{\mathbb{E}}
\newcommand{\Tr}{\mathrm{Tr}}
\newcommand{\Ek}[1]{\E\left[e_{#1}^\top P e_{#1}\right]}

\newcommand{\EE}[1]{\E\left[#1\right]}
\newcommand{\norm}[2]{\left\|#1\right\|^2_{#2}}
\newcommand{\ABKc}{\left(A + \sum_{j \in \mathscr{J}} B_{(j)}K_{j}\right)}
\newcommand{\sat}[1]{\varphi\left(#1\right)}
\usepackage{booktabs}
\usepackage{siunitx}
\usepackage{comment}
\makeatletter
\let\MYcaption\@makecaption
\makeatother
\def\BibTeX{{\rm B\kern-.05em{\sc i\kern-.025em b}\kern-.08em
    T\kern-.1667em\lower.7ex\hbox{E}\kern-.125emX}}
\begin{document}
\title{A Stochastic Tube-Based MPC Framework \\ with Hard Input Constraints}
\author{Carlo Karam, Matteo Tacchi-Bénard, and Mirko Fiacchini {}
\thanks{This work has been submitted to the IEEE for possible publication. The final published version may not be publicly accessible.}
\thanks{Matteo Tacchi-Bénard's work is supported by the French National Research Agency (ANR) under the PIVOINE project grant ANR-25-CE48-1599-01.}
\thanks{The authors are with Univ. Grenoble Alpes, Grenoble INP, CNRS, GIPSA-lab (e-mail: carlo.karam [at] gipsa-lab [dot] fr; matteo.tacchi [at] gipsa-lab [dot] fr; mirko.fiacchini [at] gipsa-lab [dot] fr).}}

\maketitle

\begin{abstract}
This work presents a stochastic tube-based model predictive control framework that guarantees hard input constraint satisfaction for linear systems subject to unbounded additive disturbances. The approach relies on a structured design of probabilistic reachable sets that explicitly incorporates actuator saturation into the error dynamics and bounds the resulting nonlinearity within a convex embedding. The proposed controller retains the computational efficiency and structural advantages of stochastic tube-based approaches while ensuring state chance constraint satisfaction alongside hard input limits. Recursive feasibility and mean-square stability are established for our scheme, and a numerical example illustrates its effectiveness.
\end{abstract}


\section{Introduction}

Model Predictive Control (MPC) has established itself as a reliable framework for the optimal control of constrained dynamical systems and is widely adopted in applications where performance must be achieved subject to state and input limitations. In the presence of uncertainty, Robust MPC (RMPC) extends the nominal formulation by guaranteeing recursive feasibility and constraint satisfaction for all admissible disturbance realizations through a worst-case formulation, most notably through the tube-based approach \cite{Mayne_2005, cannonRobust2010, Rakovic_2012}. Although such formulations provide strong deterministic guarantees, they typically induce conservatism by shrinking the admissible state and input sets, thereby limiting achievable performance \cite{Farina_2016}. Stochastic MPC (SMPC) has emerged as an alternative formulation that exploits probabilistic descriptions of disturbances \cite{mesbahStochastic2016}. By incorporating distributional information, SMPC often replaces hard constraint satisfaction with chance constraints, allowing for rare violations with user-prescribed probability levels. This relaxation generally enlarges the admissible region and improves closed-loop performance compared to worst-case designs. Moreover, the stochastic framework enables the treatment of disturbances without known bounded support \cite{mesbahStochastic2016}, such as Gaussian noise, for which deterministic robust formulations are inherently inapplicable.

Despite the substantial progress achieved in SMPC design, the treatment of input constraints remains an open problem. In contrast to state constraints, which can be meaningfully relaxed in a probabilistic sense, input constraints must be satisfied deterministically at all times as they often correspond to physical actuator saturation limits. Ensuring such guarantees in the presence of unbounded disturbances is nontrivial, as the stochastic feedback component may generate arbitrarily large control actions. Scenario and sample-based SMPC approaches can, in principle, accommodate state chance constraints together with hard input bounds by enforcing constraints over a sufficiently large set of disturbance realizations \cite{hewingScenarioBased2020}. However, these methods require access to disturbance samples, may entail significant computational effort, and typically provide guarantees only with respect to the sampled scenarios, without establishing strong closed-loop properties \cite{mesbahStochastic2016}. Analytical approaches instead rely on distributional or moment information of the disturbance. A prominent class of methods is based on affine disturbance-feedback parameterizations of the control law, initially introduced for RMPC in \cite{Goulart_2006} and adapted for chance-constrained systems in \cite{Oldewurtel_2008}. Such parameterizations enable convex reformulations of the stochastic optimal control problem by optimizing jointly over open-loop inputs and feedback gains \cite{mesbahStochastic2016}. To address hard input constraints under unbounded noise, \cite{hokayemStochastic2009, Chatterjee_2011} introduced saturation within the affine parametrization, thereby ensuring deterministic input feasibility in the absence of state constraints; this framework was subsequently extended to output-feedback settings in \cite{Hokayem_2012}. More recently, \cite{paulsonStochastic2020} generalized the approach to handle hard input bounds together with joint state constraints. While effective, these formulations are often computationally demanding and do not scale efficiently with increasing prediction horizon or system dimension, as the number of decision variables grows accordingly. Stochastic tube-based approaches offer a structurally different alternative. Inspired by robust tube MPC, they decouple nominal trajectory optimization from uncertainty propagation by computing probabilistic reachable sets (PRS) offline and tightening constraints accordingly, resulting in a deterministic and low-dimensional MPC problem online \cite{mesbahStochastic2016}. Despite being extensively investigated as in \cite{hewingStochastic2018, hewingRecursively2020, Chaouach_2022, aoStochastic2025, fiacchini_2025_recursive}, most existing stochastic tube schemes do not guarantee hard input constraints when disturbances are unbounded, as the ancillary feedback gain is fixed offline and may generate arbitrarily large inputs. Existing attempts to enforce deterministic actuator limits without assuming known robust disturbance bounds (\cite{bonzaniniTubebased2019, schluterConstraintTightening2020}) remain limited. In particular, output-feedback-based methods such as \cite{joaOutput2023, Perez_Salesa_2025} ensure that the applied input satisfies the constraints deterministically, but compromise recursive feasibility due to unrealistic input trajectories, only ensuring feasibility with a user-specified probability. More recently, a constraint-tightening SMPC framework for general nonlinear systems accommodating hard input constraints was proposed in \cite{kohlerPredictive2025}, though its application to linear systems is not fully detailed and may introduce significant conservatism.

Motivated by the above considerations, this work develops a stochastic tube-based SMPC framework for linear systems that guarantees hard input constraint satisfaction under unbounded disturbances by explicitly incorporating actuator saturation into the system dynamics. Although this modeling choice introduces a nonlinearity in the error dynamics, we embed it within an ellipsoidal PRS construction. This enables the retention of linear nominal dynamics and ancillary feedback while ensuring deterministic satisfaction of input constraints together with state chance constraint guarantees. The preliminary PRS design underlying the proposed framework can be found in \cite{Karam_2025}. In that work, we developed a constructive method for computing ellipsoidal PRS for saturated linear systems subject to unbounded additive disturbances, based on convex bounds of the saturated error dynamics and an ensuing contraction analysis. The present work builds upon that foundation and makes the following contributions:
\begin{enumerate}
    \item Integration into a complete SMPC scheme: we embed the PRS construction of \cite{Karam_2025} into a stochastic tube-based predictive control framework with hard input constraints, yielding a tractable scheme that we term Saturation-Aware SMPC (SA-SMPC).
    \item For the resulting SA-SMPC controller, we establish recursive feasibility and state chance constraint satisfaction under unbounded disturbances with known second moments. In addition, we provide rigorous closed-loop guarantees, including mean-square boundedness, both for the practically implemented nominal-cost formulation and for the ideal nonlinear formulation.
\end{enumerate}
The proposed framework guarantees the standard closed-loop properties typically required for SMPC schemes with hard input constraints, while explicitly accounting for the closed-loop structure of the error under the pre-stabilizing feedback in the offline constraint-tightening design. This enables a less conservative enforcement of hard input constraints than that provided in \cite{kohlerPredictive2025}, where the tightening is effectively based on open-loop error predictions. To the best of our knowledge, no such framework exists.

The remainder of this paper is organized as follows. Section~\ref{sec:prob-setup} introduces the problem formulation and standing assumptions. Section~\ref{sec:prs-design} presents the convex PRS synthesis for the saturated error dynamics and details the conditions under which the base bound can be refined to reduce conservatism. Section~\ref{sec:sa-smpc} develops the full SA-SMPC scheme and establishes recursive feasibility, chance constraint satisfaction, and closed-loop stability properties. Finally, Section~\ref{sec:numerical} provides numerical results validating the theoretical developments and compares the proposed method with a nominal MPC scheme as well as with the affine disturbance-feedback approach with hard input constraints of \cite{paulsonStochastic2020}. Lengthy technical proofs are deferred to Appendices~\ref{app:prop3-proof} and~\ref{app:proof-real-cost} for completeness and clarity of exposition.

\subsection{Notation}
The set of positive integers is denoted by $\NN$ and the set of nonnegative integers by $\NN_0$. For integers $a,b\in\NN$ with $a\le b$, we define $\NN_a^b := \{n\in\NN:\; a \le n \le b\}$. For a matrix $A \in \R^{n\times m}$, the notation $A_{(j)} \in \R^n$ denotes its $j$-th column, while $A_j \in \R^m$ denotes its $j$-th row. For symmetric matrix $A \in \R^{n\times m}$, $\Lambda_{\max}(A)$ is the maximum eigenvalue of $A$. For $P \succ 0$ and $r > 0$, we define the ellipsoid $\mathcal{E}(P,r) := \{e\in\R^n:\; e^\top P e \le r\}$. The Pontryagin difference of two sets $\mathcal{S}, \mathcal{T} \subseteq \R^n$ is defined as $ \mathcal{S} \ominus \mathcal{T} := \left\{s \in \mathcal{S} \mid s + t \in \mathcal{S},\ \ \forall t \in \mathcal{T} \right\}$. For a matrix $K \in \R^{m\times n}$ and a set $\mathcal{S} \subseteq \R^n$, we denote the image of $\mathcal{S}$ under the linear map $K$ by $K \mathcal{S} := \{Ks \in \R^m \mid s \in \mathcal{S}\}$. Finally, the measured state at time $k$ is denoted $x_k$, and the $i$-step-ahead state prediction made at time $k$ is denoted $x_{i|k}$, with $x_{0|k} = x_k$.
\newpage
\section{Problem Formulation}
\label{sec:prob-setup}
Consider an uncertain linear time-invariant (LTI) system whose discrete-time dynamics obey
\begin{equation}
\label{eq:lti-system}
    x_{k+1} = Ax_k + Bu_k + w_k,
\end{equation}
where $x_k \in \R^{n}$, $u_k \in \R^m$, and $w_k \in \R^n$ represent the system state, control input, and unknown process noise, respectively. The matrices $A \in \R^{n \times n}$ and $B \in \R^{n \times m}$ are fully known.

\begin{assumption}[Noise distribution]\label{ass:noise-iid}
The disturbance $w_k$ is independently and identically distributed (i.i.d.) according to an unknown distribution with zero mean and known variance, i.e.
\begin{equation*}
    \E\left[w_k\right] = \mathbf{0},  \qquad  \E\left[w_k^\top w_k \right] = W \succ 0
\end{equation*}
for all $k \in \NN_0$. 
\end{assumption}

Note that these conditions may be met by disturbance distributions with infinite support. The state constraint set $\mathcal{X} \subseteq \R^n$ and the input constraint set $\mathcal{U} \subseteq \R^m$ are defined as
\begin{equation*}\label{eq:constrain-sets}
    \mathcal{X} \coloneqq \left\{ x \in \R^n \mid H x \leq h \right\}, \qquad 
    \mathcal{U} \coloneqq \left\{ u \in \R^m \mid \|u\|_\infty \leq 1 \right\},
\end{equation*}
We impose state chance constraints, with user-defined violation probability $\varepsilon \in (0, 1)$, as \begin{equation}\label{eq:state-chance-constraints}
   \Pr\left\{ x_k \in \mathcal{X} \mid x_0 \right\} \geq 1 - \varepsilon, \qquad \forall k \in \mathbb{N},
\end{equation}
where the probability is conditioned on the initial state $x_0$, and hard input constraints
\begin{equation}\label{eq:input-constraints}
   u_k \in \mathcal{U}, \qquad \forall k \in \mathbb{N}.
\end{equation}

As usual in the SMPC framework, the aim is to solve an optimization problem at every step $k$, minimizing over a sequence of $i$-steps ahead predicted states $x_{i | k}$ and inputs $u_{i | k}$ which relate to each other by
\begin{equation}
x_{i+1|k} = A x_{i|k} + B u_{i|k} + w_{i+k}, \qquad x_{0|k} = x_k,
\label{eq:pred-states}
\end{equation}
where the realized state $x_k$ is measurable. We introduce the finite-horizon cost
\begin{equation}
    \label{eq:fh-cost}
    \mathcal{J}_N(\mathbf{x}_k, \mathbf{u}_k) \coloneqq \E \left[ \sum_{i = 0}^{N-1}\ell(x_{i | k}, u_{i | k}) + V_f(x_{N | k}) \right]
\end{equation}
to be minimized at time-step $k$, given a user-defined stage cost $\ell(\cdot,\cdot)$, terminal cost $V_f(\cdot)$, and some horizon $N \in \NN$. The resulting finite-horizon stochastic optimal control problem can be thus stated as
\begin{subequations}
       \label{eq:sfhocp}
    \begin{align}
    \min_{\mathbf{x}_k, \mathbf{u}_k} \qquad & \mathcal{J}_N(\mathbf{x}_k, \mathbf{u}_k), \label{eq:sfhocp-cost} \\
    \mathrm{s. t.} \qquad &x_{i+1 | k} = Ax_{i | k} + B u_{i | k} + w_{i + k}, \\
    &x_{0 | k} = x_k, \\
    &u_{i | k} \in \mathcal{U}, \quad i \in \NN^{N}_0, \label{eq:u-constr}\\
    &\Pr\left\{ x_{i | k} \in \mathcal{X} \,|\, x_k \right\} \geq 1 - \varepsilon, \quad i \in \NN^{N-1}_0, \label{eq:xik-constr}\\
    &\Pr\left\{ x_{N | k} \in \mathcal{X}_f \,|\, x_k \right\} \geq 1 - \varepsilon, \label{eq:xn-constr}
\end{align}
\end{subequations}
where $\mathbf{x}_k = (x_{i|k})_{i=0}^N$, $\mathbf{u}_k = (u_{i|k})_{i = 0}^N$, and $\mathcal{X}_f$ is an appropriately designed terminal constraint set.

The optimization problem above is generally intractable due to the nonconvexity and implicit nature of the chance constraints, which require evaluating multivariate probability distributions over future uncertain states. To obtain a tractable formulation, the problem is approximated through a constraint-tightening approach based on reachable sets, which ensures probabilistic feasibility while enabling online optimization of control inputs within the admissible set $\mathcal{U}$. To this end, we adopt a standard stochastic tube approach inspired by robust MPC schemes \cite{Mayne_2005}, splitting the predicted state into nominal and stochastic components
\begin{equation}\label{eq:state-decomp}
   x_{i | k} = z_{i | k} + e_{i | k},
\end{equation}
where $z_{i | k}$ denotes the nominal state, and $e_{i | k}$ represents its deviation with respect to the stochastic predicted state. The control input is parameterized as
\begin{equation}
    \label{eq:total-input}
   u_{i | k} = v_{i | k} + Ke_{i | k}.
\end{equation}
where $v_{i | k}$ is the nominal input sequence optimized online (at every step $k$), and $K$ is a fixed stabilizing feedback gain computed offline.

Because the process noise may be unbounded (see Assumption \ref{ass:noise-iid}), the error term $e_{i | k}$ may become arbitrarily large. Since $K$ is fixed, this consequently makes it difficult to bound the behavior of the pre-stabilizing control action $K e_{i | k}$, and in turn ensure hard constraints $u_{i | k} \in \mathcal{U}$ for all $k \in \NN$. To address this issue, we incorporate the input constraints \eqref{eq:input-constraints} directly into the system dynamics by introducing a saturation on the control input. The predicted state now evolves according to,
\begin{equation}\label{eq:lti-sat-system}
   x_{i+1 | k} = Ax_{i | k} + B\varphi\left(u_{i | k}\right) + w_{i + k},
\end{equation}
where $\varphi : \mathbb{R}^m \to \mathbb{R}^m$ is the element-wise saturation function 
\begin{equation}\label{eq:sat-func}
   \varphi_{j}([u_{i|k}]) = \mathrm{sign}\left([u_{i|k}]_j\right) \min\left\{\left| [u_{i|k}]_j \right|, 1 \right\},
\end{equation}
for all $j \in \NN_1^m$, with saturation bounds normalized to 1. For compactness, define
\begin{equation}
   f(e_{i | k}, v_{i | k}) \coloneqq Ae_{i | k} + B \left(\varphi(Ke_{i | k} + v_{i | k}) - v_{i | k} \right), 
\end{equation}
yielding the nominal and error dynamics
\begin{subequations}\label{eq:split-dynamics}
   \begin{align}
      z_{i+1 | k} &= Az_{i|k} + Bv_{i | k},\label{eq:nominal-dynamics}\\
      e_{i+1 | k} &= f(e_{i | k}, v_{i | k}) + w_{i + k}.\label{eq:error-dynamics}
   \end{align}
\end{subequations}
Note that, if $z_{0|k} = x_k$ then $e_{0|k} = 0$, from \eqref{eq:state-decomp}. For our development, we require the following regularity conditions.

\begin{assumption}[Regularity conditions] \label{ass:regularity-conds}
\leavevmode
\begin{enumerate}[label=(\alph*),ref=\theassumption\alph*]
    \item The matrix $A$ is Schur stable. \label{ass:ol-schur}
    \item $\|v_{i \mid k}\|_\infty \leq 1$, for all $i$, $k \in \NN_0$.
\label{ass:norm-input-cnst}
\end{enumerate}
\end{assumption} 
Condition b) is imposed in the eventual SMPC scheme, and is equivalent to enforcing \eqref{eq:input-constraints} on the nominal system (given normalized saturation bounds). 

The error dynamics in \eqref{eq:error-dynamics} provide the basis for recasting the chance constraints into deterministic, tightened constraints on the nominal state vector. This process is detailed in the following section. 


\section{Probabilistic Reachable Set Construction}
\label{sec:prs-design}
This section details the methodology used to construct the sets required to make the probabilistic constraints in \eqref{eq:xik-constr} and \eqref{eq:xn-constr} tractable. We begin by formalizing the core set-theoretic concepts employed in that analysis.

\begin{definition}[Probabilistic Reachable Set]
   For a given $k \in \NN$, a set $\mathcal{R}^\varepsilon_k$ is a horizon $k$ Probabilistic Reachable Set (PRS) with violation probability $\varepsilon \in [0, 1]$ for a stochastic process $\left\{e_i\right\}_{i \in \mathbb{N}} \in \R^n$ if 
   \begin{equation*}
      \Pr\left\{ e_k \in \mathcal{R}^\varepsilon_k \mid e_0 = 0 \right\} \geq 1 - \varepsilon.
   \end{equation*}
\end{definition}

\begin{definition}[Probabilistic Ultimate Bound, \cite{kofmanProbabilistic2012}]

   A set $\mathcal{R}^\varepsilon$ is a Probabilistic Ultimate Bound (PUB) with violation probability $\varepsilon \in [0, 1]$ for a stochastic process $\left\{e_k\right\}_{k \in \mathbb{N}} \in \mathbb{R}^n$ if, for any initial condition $e_0$, there exists a finite time $t(e_0)$ such that 
   \begin{equation*}
      \Pr\left\{e_k \in \mathcal{R}^\varepsilon \right\} \geq 1 - \varepsilon,  
      \quad \forall k \geq t(e_0).
   \end{equation*}
\end{definition}

While the preceding definitions are valid for sets of arbitrary geometry, the subsequent development focuses on their ellipsoidal representations. Specifically, we develop a recursive procedure to compute ellipsoidal PRS for the $i$-step ahead predicted error $e_{i \mid k}$ conditioned on $e_{0 \mid k} = 0$.

\subsection{Convex Bounds of Saturated Functions}\label{sec:co-bounds-sat}
Reachability analysis typically relies on set propagation methods that require explicit knowledge of the system's dynamics to compute reachable sets under certain favorable conditions \cite{blanchini2008set}, \cite{kolmanovsky1998theory}. Such conditions break down for the nonlinear error dynamics given above due to the hard nonlinearity introduced by the saturation function $\varphi(\cdot)$ and the dependence on the nominal input $v$. As a result, standard recursive set computation techniques are hardly applicable. To address this issue, we aim to bound $f(e_{i \mid k}, v_{i \mid k})$ within a well-defined convex set (as in \cite{fiacchiniQuadratic2012}), which allows for the computation of outer approximations of the true reachable sets of the error dynamics.

\begin{definition}[Support Function]
   For a nonempty convex set $C \subseteq \mathbb{R}^n$, the support function evaluated at $\eta \in \mathbb{R}^n$ is defined as $\delta_C(\eta) = \sup_{x \in C} \, \eta^\top x$.
\end{definition}

The properties inherent to support functions lead to the following set inclusion theorem.

\begin{theorem}[\cite{rockafellarConvex2015}]\label{th:support-inclusion}
   \emph{For nonempty, convex, and closed sets $C_1$ and $C_2$, the inclusion $C_1 \subseteq C_2$ holds if and only if $\delta_{C_1}(\eta) \leq \delta_{C_2}(\eta)$ for every $\eta \in \mathbb{R}^n$.}
\end{theorem}

From this result, we present the subsequent theorem, valid for any $e \in \mathbb{R}^n.$

\begin{theorem}\label{th:sat-incl-co}
   \emph{For any $v$ satisfying $\left\| v \right\|_\infty \leq 1$, the function $f(e, v) = Ae + B \left(\varphi(Ke + v) - v\right)$ is contained within the set $F(e)$, where
   \begin{equation}\label{eq:co}
      F(e) = \mathrm{co}\left(\left\{  \Big(A + \sum_{j \in \mathscr{J}} B_{(j)}K_j \Big) e  \right\}_{\mathscr{J} \subseteq \mathbb{N}^m_1} \right).
   \end{equation}
} 
\end{theorem}

\proof Consider first the scalar case $m = 1$ and assume $Ke \geq 0$. Given the constraint $-1 \leq v \leq 1$, it follows that $-1 \leq Ke + v \leq 1 + Ke$, and consequently
\begin{equation*}
   \varphi(Ke + v) = 
   \begin{cases}
      Ke + v & \text{for } Ke + v \leq 1, \\
      1 & \text{for } Ke + v \geq 1.  \\
   \end{cases}
\end{equation*}
This leads to
\begin{equation*}
   \varphi(Ke + v) - v = 
   \begin{cases}
      Ke & \text{for } Ke + v \leq 1, \\
      1 - v& \text{for } Ke + v \geq 1.  \\
   \end{cases}
\end{equation*}
Since $Ke$ is unbounded above and $1 - v \geq 0$, the inequality $0 \leq \varphi(Ke + v) - v \leq Ke$ holds. A similar argument for $Ke < 0$ yields $Ke \leq \varphi(Ke + v) - v \leq 0$. Generalizing to arbitrary $m$, for any index subset $\mathscr{J} \subseteq \mathbb{N}_{1}^m$ and $j \in \mathscr{J}$, we obtain
\begin{equation*}
\begin{cases}
    \hphantom{K_je}0 \leq \varphi \left( K_j e + v_j \right) - v_j \leq K_{j}e & \quad \text{when } K_j e \geq 0, \\
    \hphantom{0}K_{j}e \leq \varphi \left( K_je + v_j \right) - v_j \leq 0 & \quad \text{when } K_j e < 0.
\end{cases}
\end{equation*}
By leveraging the support function, we deduce that for every $\eta \in \mathbb{R}^n$ and every $j \in \mathscr{J}$,
\begin{equation*}
   \eta^\top B_{(j)} \left( \varphi(K_je + v_j) - v_j \right) \in 
   \mathrm{co}\left\{ 0, \eta^\top B_{(j)} K_{j} e \right\}
   \subseteq \mathbb{R}.
\end{equation*}
Equivalently, the bounds
\begin{align*}
   \min\left\{ 0, \eta^\top B_{(j)} K_{j} e \right\} &\leq 
   \eta^\top B_{(j)} \left( \varphi(K_je + v_j) - v_j \right) \leq \max\left\{ 0, \eta^\top B_{(j)} K_{j} e \right\},
\end{align*}
are satisfied. Therefore, for any $\eta \in \mathbb{R}^n$, any $v$ with $\|v\|_{\infty} \leq  1$, and any $e \in \R^n$, there exists an index set $\mathscr{J}(e, \eta) \subseteq \mathbb{N}_1^m$ such that
\begin{align*}
   \begin{split}
      \eta^\top f(e, v) &= \eta^\top Ae + \sum_{j \in \NN_1^m} \eta^\top B_{(j)} \left(\varphi(K_j e + v_j) - v_j \right) \\
      &\leq \eta^\top Ae + \sum_{j \in \mathscr{J}(e, \eta)} \eta^\top B_{(j)} \left(\varphi(K_j e + v_j) - v_j \right).
   \end{split}
\end{align*}
Application of Theorem \ref{th:support-inclusion} confirms that $f(e, v) \in F(e)$.
\endproof

Theorem \ref{th:sat-incl-co} establishes that for every $e \in \R^n$ and every $\left\| v \right\|_\infty \leq 1$, the value $f(e, v)$ lies within the polytope $F(e)$, whose vertices correspond to the linear combinations $\ABKc e$ for all $\mathscr{J} \subseteq \mathbb{N}_1^m$. These vertices correspond to realizations of different saturation scenarios: full, partial, and no saturation \cite{Alamo_2006}.

The following corollaries are direct consequences of Theorem \ref{th:sat-incl-co}, derived via convexity arguments. To simplify notation, for every $ \mathscr{J} \subseteq \NN_1^m$ define $A_{K} (\mathscr{J}) = \ABKc$.

\begin{corollary}\label{cor:2norm-bound}
   \emph{If Assumption \ref{ass:regularity-conds} is satisfied, then the inequality       \begin{align}\label{eq:max-ineq}
         \begin{split}
      f(e,v)^\top P f(e,v) \leq \max_{\mathscr{J} \subseteq \NN_1^m} \,\, e^\top A_{K} (\mathscr{J})^\top P A_{K} (\mathscr{J}) e,
         \end{split}
   \end{align}
   holds for every $e \in \R^n$.}
\end{corollary}

\proof This assertion follows directly from Theorem \ref{th:sat-incl-co}, by convexity of $\R^n \ni f \longmapsto f^\top P f$. 
\endproof

These results indicate that the nonlinear dynamics $f(e, v)$ are bounded by the behavior of the linear systems $A_{K} (\mathscr{J}) e$ for all subsets $\mathscr{J} \subseteq \mathbb{N}_1^m$ and all admissible $v$. Consequently, bounds derived for $A_{K} (\mathscr{J}) e$ are also valid for the saturated system.

\subsection{Base PRS Synthesis}
Utilizing the results from Section \ref{sec:co-bounds-sat}, we now construct a PRS for the error state $e_{i | k} \in \R^n$.

\begin{proposition}\label{prop:base-exp-bound}
   \emph{Let Assumptions \ref{ass:noise-iid} and \ref{ass:regularity-conds} hold, and consider a symmetric positive definite matrix $P = P^\top \succ 0$ that satisfies
\begin{equation}\label{eq:contractivity-lmi}
{A_K}(\J)^\top \, P \, {A_K}(\J) \preceq \lambda P, \qquad \forall \J \subseteq \NN_1^m 
\end{equation}
for some contraction rate $\lambda \in [0, 1)$. Then, for all $i, k \in \mathbb{N}$, the bound
\begin{equation}\label{eq:recursive-exp-bound}
   \E\left[e_{i+1|k}^\top P e_{i + 1 | k} \right] \leq \lambda \E\left[e_{i | k}^\top P e_{i | k} \right] + \mathrm{Tr}\left( PW \right)
\end{equation}
holds for system (\ref{eq:error-dynamics}) and
\begin{equation}\label{eq:exp-bound}
\begin{split}
\E\left[ e_{i | k}^\top P e_{i | k} \right] \leq \frac{1 - \lambda^i}{1 - \lambda} \mathrm{Tr}(PW),
\end{split}
\end{equation}
if $e_{0 | k} = 0$.}
\end{proposition}

\proof If \eqref{eq:contractivity-lmi} is satisfied, Corollary \ref{cor:2norm-bound} can be applied, yielding 
\begin{equation*}\label{eq:ellipsoid-bound}
   f(e_{i | k}, v_{i | k})^\top P f(e_{i | k}, v_{i | k}) \leq \lambda e^\top_{i | k} P e_{i | k}.
\end{equation*}
for all $e_{i|k} \in \R^n$. This inequality and Assumption \ref{ass:noise-iid} lead to the following relationship:
\begin{align*}
      \E [ (f(e_{i | k}, v_{i | k}) + w_{i+k} )^\top \! P \! & \left(f(e_{i | k}, v_{i | k}) + w_{i+k}\right) ] 
      = \E\left[ f(e_{i | k}, v_{i | k})^\top P f(e_{i | k}, v_{i | k}) \right] + \E\left[w_{i+k}^\top P w_{i+k} \right] \\ 
      & = \E\left[ f(e_{i | k}, v_{i | k})^\top P f(e_{i | k}, v_{i | k}) \right] + \mathrm{Tr} \left(P W \right) \leq \lambda\E\left[ e_{i | k}^\top P e_{i | k} \right] + \mathrm{Tr} \left(P W \right).
\end{align*}
Considering the error dynamics \eqref{eq:error-dynamics}, we obtain \eqref{eq:recursive-exp-bound}. The remainder of the proof proceeds by induction. For $i = 0$, $\E\left[ e_{0|k}^\top P e_{0 | k} \right] = 0$ by assumption. For $i = 1$,
\begin{align*}
    \E\left[ e_{1 | k}^\top P e_{1 | k} \right] \leq \lambda \E\left[e_{0 | k}^\top P e_{0 | k} \right] + \mathrm{Tr}(PW) \leq \mathrm{Tr}(PW).
\end{align*}
For $i = t$, one gets
\begin{align*}\label{eq:geom-series}
\begin{split}
    \E\left[ e_{t|k}^\top P e_{t|k} \right] \leq \lambda \E\left[e_{t-1 | k}^\top P e_{t-1 | k} \right] + \mathrm{Tr}(PW) \leq \sum_{j = 0}^{t-1} \lambda^{j} \mathrm{Tr}(PW). 
\end{split}
\end{align*}
The sequence of upper bounds for $\E\left[e_{i|k}^\top P e_{i|k}\right]$ forms a geometric series with ratio $\lambda$. Thus, inequality (\ref{eq:exp-bound}) is valid for any contraction rate $0 \leq \lambda < 1$. \endproof

\begin{remark}
    In the preceding proof, Assumption \ref{ass:noise-iid} was used to nullify the cross-terms. Should the analysis be extended to incorporate biased or correlated disturbances, the terms $\E \left[w_{i+k} \right]^\top P \E \left[w_{i+k} \right]$ and $\E \left[ f \left( e_{i | k}, v_{i | k} \right)^\top P w_{i+k} \right]$ would need to be introduced, respectively. Accommodating such terms within the present framework is possible by assuming known bounds on the means and covariance matrices, and following the approach detailed in \cite{fiacchiniProbabilistic2021}. 
\end{remark}

Combining the previous result with Markov's inequality yields the following lemma.

\begin{lemma}\label{lem:markov-reach}
   \emph{The sequence of sets defined by
   \begin{equation}\label{eq:k-step-markov}
      \mathcal{R}^\varepsilon_i(\lambda) = \left\{e \in \R^n \; \middle| \; e^\top P e \leq 
      \frac{1 - \lambda^i}{\varepsilon \left( 1 - \lambda \right)}\mathrm{Tr}\left( PW \right) \right\},
   \end{equation}
forms a sequence of Markov-based ellipsoidal PRS for violation probability $\varepsilon$.}
\end{lemma}

\proof Applying Markov's inequality to the random variable $e_{k}^\top P e_{i|k}$ gives
\begin{equation*}
   \Pr\left\{ e_{i|k}^\top P e_{i|k} \geq r \right\} \leq \frac{\E\left[ e_{i|k}^\top P e_{i|k} \right]}{r}.
\end{equation*}
Plugging in the bound in \eqref{eq:exp-bound} into the right hand side yields for $\displaystyle r = \frac{1 - \lambda^i}{\varepsilon \left( 1 - \lambda \right)} \Tr(PW)$
\begin{align*}
   \Pr\left\{ e_{i|k}^\top P e_{i|k} \leq r \right\} &\geq 1 - \frac{1 - \lambda^i}{r\left(1 - \lambda\right)} \mathrm{Tr}(PW)  = 1 - \varepsilon. \tag*{\QED}
\end{align*}

Taking the limit of the Markov-based PRS $\mathcal{R}^\varepsilon_i (\lambda)$ as $i \to \infty$ produces the following Markov-based ellipsoidal PUB set.

\begin{lemma}\label{lem:markov-bound}
   \emph{The set defined by
   \begin{equation}
      \mathcal{R}^\varepsilon(\lambda) = \left\{e  \in \R^n  \; \middle | \; e^\top P e \leq 
      \frac{1}{\varepsilon \left( 1 - \lambda \right)}\mathrm{Tr}\left( PW \right) + \zeta
      \right\},
   \end{equation}
with $\zeta > 0$, is a Markov-based ellipsoidal PUB for violation probability $\varepsilon$.} 
\end{lemma}

\proof
For any initial condition $e_{0|k} \in \R^n$, the recursion in~\eqref{eq:recursive-exp-bound} implies
\begin{equation*}
   \lim_{i \to \infty} \Ek{i|k} \leq \lim_{i \to \infty} \left( \lambda^i \Ek{0|k} + \frac{1 - \lambda^i}{1 - \lambda} \mathrm{Tr}(PW) \right)
   = \frac{1}{1 - \lambda}\mathrm{Tr}\left( PW \right). 
\end{equation*}
Hence, by definition of limits, for every $\varepsilon,\,\zeta > 0$ and every $e_{0|k} \in \R^n$, there exists $t(e_{0|k}) \in \NN$ such that
\begin{equation*}
    \Ek{i|k} \leq \frac{1}{1 - \lambda} \Tr(PW) + \varepsilon \zeta, \qquad \forall i \geq t(e_{0|k}).
\end{equation*}
Applying Markov's inequality yields
$$
\Pr\left\{e_{i|k}^\top P e_{i|k} \leq \frac{1}{\varepsilon} \left( \frac{1}{1 - \lambda} \Tr(PW) + \varepsilon \zeta \right) \right\} \geq 1 - \varepsilon, \qquad \forall i \geq t(e_{0|k}),
$$
which proves the result. 
\endproof

Note that for any stabilizing gain $K$, condition (\ref{eq:contractivity-lmi}) also requires that the quadratic function satisfies the Lyapunov condition with the same rate $\lambda$ i.e. 
\begin{equation}\label{eq:ol-contraction-lmi}
A^\top P A \preceq \lambda P.
\end{equation}
The resulting contraction rate $\lambda$ therefore characterizes the open-loop or fully saturated dynamics. This bound, equivalent to the one in \cite{kohlerPredictive2025} for a linearized system, is inherently conservative, as it ignores the stronger contraction achievable near the origin under non-saturated, closed-loop operation. The following section introduces a refined bound to address this limitation.

\subsection{Conditional Refinement of Base PRS}
\label{subsec:tight-prs}
The Markov-based PRS and PUB presented in Lemmas \ref{lem:markov-reach} and \ref{lem:markov-bound} exhibit conservatism, as they rely on a contraction rate $\lambda$ that does not fully characterize the system dynamics. In practice, the system often evolves in a region where the control input $Ke + v$ does not saturate, and the dynamics are linear. This region is defined below.

\begin{definition}[Saturation Region of Linearity\cite{tarbouriechStability2011}]
   The region of linearity for system (\ref{eq:error-dynamics}), denoted $\mathscr{R}_L$, is the set of all error states $e$ for which $\varphi\left(Ke + v\right) = Ke + v$.
\end{definition}

Inside $\mathscr{R}_L$, the predicted error dynamics reduce to $e_{i+1 | k} = (A+ BK)e_{i|k} + w_{i+k}$. The following corollary, which follows from Theorem \ref{th:sat-incl-co}, Proposition \ref{prop:base-exp-bound}, and the preceding definition, is presented.

\begin{corollary}\label{cor:lambdal}
   \emph{Let Assumptions \ref{ass:noise-iid} and \ref{ass:regularity-conds} hold and $P = P^\top \succ 0$ be a matrix satisfying the Lyapunov inequality (\ref{eq:contractivity-lmi}) with contraction rate $\lambda \in [0, 1)$. Suppose the gain $K \in \mathbb{R}^{n \times m}$ and a scalar $\lambda_L \in \mathbb{R}$ satisfy
   \begin{subequations}
      \begin{align}
   &      \lambda_L \leq \lambda, \label{eq:fast-contraction}\\
   &   (A + BK)^\top P (A + BK) \preceq \lambda_L P
 \label{eq:schur-lmi},
      \end{align}
      \label{eq:find-lambdal}
   \end{subequations}
   Then, the inequality
   \begin{equation}\label{eq:lambdal-bound}
      \Ek{k+1} \leq \lambda_L \Ek{k} + \Tr(PW),
   \end{equation}
holds for every $k \in \NN_0$ such that $e_k \in \mathscr{R}_L$.
}
\end{corollary}

\proof Condition \eqref{eq:schur-lmi} establishes that the contraction rate $\lambda_L$, applicable within $\mathscr{R}_L$, is less than or equal to $\lambda$. 
Thus $e_{i|k}^\top P e_{i|k} = \Ek{k}$ for all $e_{i|k} \in \mathscr{R}_L$ and inequality (\ref{eq:recursive-exp-bound}) holds with $\lambda_L$ in spite of $\lambda$, resulting in (\ref{eq:lambdal-bound}).
\endproof
To harmonize the open-loop contraction from \eqref{eq:contractivity-lmi} with the local closed-loop contraction valid in $\mathscr{R}_L$, we impose the following compatibility condition.
\begin{assumption}\label{ass:stable-cl}
   The matrices $K$ and $P$, along with parameters $\lambda$ and $\lambda_L$, are such that conditions \eqref{eq:contractivity-lmi} and \eqref{eq:find-lambdal} are satisfied.
\end{assumption}
\begin{remark}
Assumption \ref{ass:stable-cl} requires the existence of a common Lyapunov function $P$ and a feedback gain $K$ that simultaneously satisfy the open-loop and closed-loop contraction conditions. In particular, the aim is to find a solution with  $\lambda_L \leq \lambda$, ensuring the closed-loop dynamics contract at least as fast as the open-loop dynamics. A feasible pair $(P, K)$ can be synthesized by solving a convex optimization problem subject to the constraints \eqref{eq:contractivity-lmi} and \eqref{eq:find-lambdal}. Through the change of variables $X = P^{-1}$ and $Y = KX$ and a Schur complement formulation, we get 
\begin{subequations}\label{eq:opt-params}
    \begin{align}
        \min_{X, Y, \lambda_L, \lambda} \quad & \lambda + \lambda_L \\
        \text{s.t.} \quad \;\;\,& \lambda_L \leq \lambda \\
                         & \begin{bmatrix}
                             \lambda X & \left(AX + B_{(j)}Y_j\right)^\top \\
                             AX + B_{(j)}Y_j & X
                         \end{bmatrix} \succeq 0, \qquad \forall j \in  \mathscr{J}, \quad \forall  \mathscr{J} \subseteq \NN_1^m, \label{eq:schur-partial}\\
                         & \begin{bmatrix}
                             \lambda_L X & (AX + BY)^\top \\
                             AX + BY  &  X
                         \end{bmatrix} \succeq 0,
    \end{align}
\end{subequations}
where (\ref{eq:schur-partial}) with $\J = \{\emptyset\}$ is equivalent to (\ref{eq:ol-contraction-lmi}). The controller gain is then recovered as $K = Y X^{-1}$. Problem~\eqref{eq:opt-params} can be solved to global optimality using the moment-SoS hierarchy for optimization with polynomial matrix inequality constraints~\cite{Henrion_2006}, but such a development is deferred to a later work. Alternatively, the parameters $\lambda$ and $\lambda_L$ can be optimized via a 2-level bisection search to find the fastest achievable contraction rates. From a practical perspective, it is conceivable that incorporating additional performance criteria may lead to a better overall scheme; e.g. regularization terms involving $KP$, or objectives that promote smaller ellipsoids (such as minimizing $- \log \det Q$). Since we do not develop a dedicated convex formulation that guarantees optimality, we omit such extensions. Note that problem \eqref{eq:opt-params} is always feasible under Assumption \ref{ass:regularity-conds}.
\end{remark}

We now have two distinct bounds on the evolution of $\mathbb{E}\left[e_{i|k}^\top P e_{i|k}\right]$: one that is globally valid but conservative, and another that is tighter, valid only locally within the region of linearity. It is expected that the true system behavior falls somewhere between these two extremes. To capture this behavior more effectively, we combine the worst-case and best-case bounds probabilistically, with weights reflecting the likelihood of the error state’s position in the state space. The following theorem helps formalize this approach within the chosen ellipsoidal set framework.

\begin{theorem}[Ellipsoidal Region of Linearity~\cite{tarbouriechStability2011}]\label{th:ellipsoidal-Rl}
   \emph{The largest ellipsoid with shape matrix $P$ contained within the region of linearity, $\mathscr{E}(P, r_L) \subseteq \mathscr{R}_L$, has a scaling parameter $r_L$ given by
   \begin{equation}\label{eq:rL}
      r_L = \min_{1 \leq j \leq m} \frac{\left(1 - v_j \right)^2}{K_{j} P^{-1} K_{j}^\top}.
   \end{equation}}
\end{theorem}
This ellipsoid will be used to determine whether error states reside in $\mathscr{R}_L$. Specifically, if $e^\top P e \leq r_L$, then $\varphi\left(Ke + v\right) = Ke + v$.

\begin{proposition}\label{prop:balanced-bound}
   \emph{Let Assumptions \ref{ass:noise-iid}, \ref{ass:regularity-conds} and \ref{ass:stable-cl} hold. Then, the inequality
   \begin{equation}\label{eq:exp-bound-tight}
      \begin{split}
         \Ek{i+1|k} \leq \lambda_L \Ek{i|k} + \frac{\lambda - \lambda_L}{r_L} \E^2\left[e_{i|k}^\top P e_{i|k} \right] + \Tr(PW),
      \end{split}
   \end{equation}
is satisfied for all $i, \, k \in \NN_1^m$ for system (\ref{eq:error-dynamics}) with $e_{0|k} = 0$.}
\end{proposition}

\proof 
Distinct contraction rates $\lambda_L$ and $\lambda$ apply when $e_{i|k} \in \mathscr{R}_L$ and $e_{i|k} \notin \mathscr{R}_L$, respectively. The state space $\R^n$ is partitioned into $\mathcal{S}_1 = \set{e \in \R^n : e^\top P e \leq r_L}$ and $\mathcal{S}_2 = \set{e \in \R^n : e^\top P e > r_L}$ with $r_L$ as in (\ref{eq:rL}), where, by construction, $\mathcal{S}_1 \subseteq \mathscr{R}_L$. The law of total expectation gives
\begin{align*}
   \begin{split}
   \E\left[ e_{i+1|k}^\top P e_{i + 1 | k} \right] &= \sum_{j = 1}^2 \E\left[e_{i+1|k}^\top P e_{i + 1 | k} \mid 
   e_{i|k} \in \mathcal{S}_j\right] \cdot \Pr\left\{e_{i|k} \in \mathcal{S}_j\right\},
    \end{split}
\end{align*}
Substituting with the bounds \eqref{eq:recursive-exp-bound} and \eqref{eq:lambdal-bound}, we get
\begin{align*}
    \begin{split}
   \E\left[ e_{i+1|k}^\top P e_{i + 1 | k} \right] \leq &\left( \lambda_L \Ek{i+1|k} + \Tr(PW) \right) \cdot \Pr\left\{e_{i|k} \in \mathcal{S}_1 \right\} \\
   & \quad + \left( \lambda \Ek{i+1|k} + \Tr(PW) \right) \cdot \Pr\left\{e_{i|k} \in \mathcal{S}_2 \right\}.
   \end{split}
\end{align*}
Since $\Pr \left\{e_{i|k} \in \mathcal{S}_2 \right\} = 1 - \Pr \left\{e_{i|k} \in \mathcal{S}_1 \right\}$, we have 
\begin{equation}
\label{eq:prob-exp-bound}
   \begin{split}
      \E\left[e_{i+1|k}^\top P e_{i + 1 | k}\right] \leq \lambda\Ek{i+1|k} - \left( \lambda - \lambda_L \right)\,\Ek{i+1|k} \cdot \Pr\left\{ e_{i|k} \in \mathcal{S}_1 \right\} + \Tr(PW).
   \end{split}
\end{equation}
Finally, substituting with the Markov bound on $\Pr \left\{e_{i|k} \in \mathcal{S}_1 \right\}$, yields
\begin{subequations}\label{eq:exp-bounds-tight}
   \begin{align}
      \begin{split}
         \E \left[e_{i+1|k}^\top P e_{i + 1 | k} \right] &\leq \lambda\E \left[e_{i|k}^\top P e_{i| k} \right] + \left( \lambda_L - \lambda \right) \E\left[e_{i|k}^\top P e_{i| k} \right] \left( 1 - \frac{\E \left[e_{i|k}^\top P e_{i| k} \right]}{r_L} \right) + \Tr(PW),
         \end{split} \nonumber\\
         &\leq \lambda_L \E\left[e_{i|k}^\top P e_{i| k} \right] + \frac{\lambda - \lambda_L}{r_L} \E^2\left[e_{i|k}^\top P e_{i| k} \right] + \Tr(PW). \nonumber
   \end{align} 
\end{subequations}

Therefore, inequality (\ref{eq:exp-bound-tight}) holds for any pair of contraction rates $\left(\lambda_L, \lambda\right)$ satisfying $0 \leq \lambda_L \leq \lambda < 1$.
\endproof

As $r_L$ increases, the bound approaches the linear case. This is consistent with the intuition that when $\mathscr{R}_L$ expands infinitely, nearly all error terms should fall within it. The relation \eqref{eq:exp-bound-tight} above further indicates that the true contraction rate of the system varies between its linear and saturated (or open-loop) bounds, depending on the magnitude of $r_L$. The following analysis aims at exploring this transition and derives a bound that captures this effective contraction behavior in a more interpretable form.

\begin{proposition}\label{prop:tighter-bound-on-exp}
   \emph{Let Assumptions \ref{ass:noise-iid}, \ref{ass:regularity-conds} and \ref{ass:stable-cl} hold, and consider the bound given by (\ref{eq:exp-bound-tight}).
   If the condition
   \begin{equation}\label{eq:rl}
    \frac{1}{1 - \lambda} \Tr(PW) < r_L,
   \end{equation}
is satisfied, then
   \begin{equation}\label{eq:exp-bound-bar-cf}
      \Ek{i|k} \leq \frac{1 - \left(\bar\lambda^\ast\right)^i}{1 - \bar\lambda} \Tr(PW),
   \end{equation}
holds, where $\bar\lambda^\ast$ is the unique value in $\left[ \lambda_L,\; \lambda \right]$ satisfying
   \begin{equation}\label{eq:lambdastar}
         \frac{1}{1 - \bar\lambda^\ast} \Tr(PW) = \frac{\bar\lambda^\ast - \lambda_L} {\lambda - \lambda_L} r_L.
   \end{equation}}
\end{proposition}

\proof The proof is given in Appendix \ref{app:prop3-proof}.

The result below follows directly from Proposition \ref{prop:tighter-bound-on-exp}, see its proof.
\begin{corollary}\label{cor:finalbound}
    \emph{Let Assumptions \ref{ass:noise-iid}, \ref{ass:regularity-conds} and \ref{ass:stable-cl} hold, and consider the bound given by (\ref{eq:exp-bound-tight}). If condition (\ref{eq:rl}) holds with $\bar{\lambda}^* \in [\lambda_L,\, \lambda]$ satisfying (\ref{eq:lambdastar})}, then
\begin{equation}
\E\left[{e_{i|k}^\top P e_{i|k}}\right] \leq \frac{1}{1-\bar\lambda^*}\Tr(PW) \quad \implies \quad \E\left[{e_{i+j|k}^\top P e_{i+j|k}}\right] \leq \frac{1}{1-\bar\lambda^*}\Tr(PW), \ \ \ \forall j \in \NN     
\end{equation}
for any $i \in \NN$.
\end{corollary}

Proposition \ref{prop:tighter-bound-on-exp} and Corollary \ref{cor:finalbound} establish the existence of an effective contraction rate $\bar\lambda^\ast$ provided the set $\mathscr{E}(P, r_L)$ is sufficiently large. Such an effective rate enables the derivation of a deterministic bound on $\mathbb{E}[e_{i|k}^\top P e_{i|k}]$ tighter than the one provided in \eqref{eq:lambdal-bound}. It can be shown that as $r_L$ grows infinitely large, $\bar\lambda^\ast \to \lambda_L$. Conversely, when $r_L$ falls below the critical threshold $(1 - \lambda)^{-1} \Tr(PW)$, the existence of such a rate can neither be confirmed nor excluded, but it becomes computationally inaccessible by using the Markov inequality. In such a case, the loose Markov inequality
\begin{equation}
   \Pr \left\{e_{i|k}^\top P e_{i|k} \leq r_L \right\} \geq 1 - \frac{\Ek{i|k}}{r_L},
\end{equation}
fails to provide a useful bound on the error state's magnitude. Consequently, if (\ref{eq:rl}) does not hold, one must resort to the more conservative bound in \eqref{eq:exp-bound}. Formally, we define

\begin{equation}\label{eq:hatlambda}
    \hat\lambda = \begin{cases}
        \bar\lambda^\ast &\qquad \text{if } \frac{1}{1 - \lambda} \Tr(PW) < r_L, \\
        \lambda &\qquad \text{if } \frac{1}{1 - \lambda} \Tr(PW) \geq r_L,
    \end{cases}
\end{equation}
where $\bar{\lambda}^\ast < \lambda$ is the solution to (\ref{eq:lambdastar}).

\begin{lemma}\label{lem:markov-reach-tight}
   \emph{The sequence of sets defined by
\begin{equation}
\mathcal{R}^\varepsilon_i(\hat\lambda) = \left\{e  \in \R^n \; \middle | \; e^\top P e \leq 
\frac{1 - \hat\lambda^i} {\varepsilon \left( 1 - \hat\lambda \right)}\mathrm{Tr}\left( PW \right) \right\},
\end{equation}
forms a sequence of Markov-based ellipsoidal PRS for violation probability $\varepsilon$, with $\hat\lambda$ defined in (\ref{eq:hatlambda}).}
\end{lemma}

\proof This follows from Proposition \ref{prop:tighter-bound-on-exp} and application of Markov's inequality to the random variable $e_{i|k}^\top P e_{i|k}$.
\endproof
Taking the limit of the Markov-based PRS $\mathcal{R}^\varepsilon_i(\hat\lambda)$ as $i \to \infty$ yields the following Markov-based PUB.

\begin{lemma}\label{lem:markov-bound-tight}
\emph{The set defined by
\begin{equation}\mathcal{R}^\varepsilon(\hat{\lambda}) = \left\{e \in \R^n \; \middle | \; e^\top P e \leq \frac{1}{\varepsilon \left( 1 - \hat\lambda \right)}\mathrm{Tr}\left( PW \right) \right\},
\end{equation}
is a Markov-based ellipsoidal PUB for violation probability $\varepsilon$, with $\hat\lambda$ as defined in (\ref{eq:hatlambda}).
}
\end{lemma}

These sets are analogous to those in Lemmas \ref{lem:markov-reach} and \ref{lem:markov-bound}, using a bound between \eqref{eq:exp-bound-bar-cf} and the conservative bound \eqref{eq:exp-bound} when the effective contraction rate $\bar\lambda^\ast$ is inaccessible. These sets will be used for constraint tightening in the full SMPC scheme detailed in the next section.

\section{Saturation-Aware Stochastic MPC}
\label{sec:sa-smpc}

The different ingredients of the proposed SMPC scheme, which we name Saturation-Aware SMPC (SA-SMPC) will be detailed in this section. For the development below, we will be operating under the following assumption.

\begin{assumption}
\label{ass:vss-lambdastar}
    The nominal control input $v_{i|k}$ is such that
    $$
        v_{i|k} \in \mathcal{V} \coloneqq \left\{v \in \R^m \mid \|v\|_\infty \leq v_\mathrm{ss} < 1 \right\},
    $$
    for all $i, \,k$, for which we can compute $\displaystyle r_L = \min_{1 \leq j \leq m} \frac{(1 - v_\mathrm{ss})^2}{K_j P^{-1} K_j^\top}$ and corresponding $\bar\lambda^\ast$ satisfying Proposition \ref{prop:tighter-bound-on-exp}, valid for all $i, \, k$.
\end{assumption}

Under this assumption, we can reformulate the state chance constraints in \eqref{eq:state-chance-constraints} as deterministic under the standard tightening below.
\begin{enumerate}[label=\alph*)]
    \item \emph{State Constraints:} we impose $z_{i|k} \in \mathcal{Z}_i = \mathcal{X} \ominus \mathcal{R}_i^\varepsilon \left(\hat\lambda\right)$ for all $i \in \NN_0^{N-1}$.
    \item \emph{Terminal Constraint:} we impose $z_{N \mid k} \in \mathcal{Z}_f \subseteq \Big( \mathcal{X} \ominus \mathcal{R}^\varepsilon \left(\hat\lambda \right) \Big) \cap \mathcal{X}_K$, where $\mathcal{Z}_f$ is a forward invariant set for the autonomous system $z^+ = (A + B K_f) z$ where $K_f$ defines the terminal control law, and $\mathcal{X}_K$ is defined as
$\mathcal{X}_K = \{z \in \mathcal{X} \, \mid \, K_f z \in \mathcal{V} \}.$
\end{enumerate}

\subsection{Cost Function and State Initialization}
Recall that the finite-horizon stochastic optimal control problem introduced in \eqref{eq:sfhocp} seeks to minimize the expected value of a cumulative stage cost and terminal cost evaluated along the closed-loop trajectories. For practical implementation, we consider quadratic stage and terminal costs typical in MPC formulations
\begin{equation}\label{eq:stage-terminal-costs}
\ell\left(x, u\right) := \| x\|^2_Q + \| u\|^2_R, \qquad V_f\left(x\right) := \|x\|^2_{S},
\end{equation}
with $Q \succeq 0$, $R \succ 0$, and $S \succ 0$ satisfying
\begin{equation}
    \label{eq:terminal-lyap}
        (A + BK_f)^\top S (A + BK_f) - S \preceq -Q - K_f^\top R K_f, 
\end{equation}
Naturally, a possible choice for the pair $(S, K_f)$ is the solution of the LQR problem for which \eqref{eq:terminal-lyap} is satisfied with equality.

A theoretically appealing SMPC formulation is the so-called indirect-feedback framework, in which the real predicted state and input trajectories enter the cost function, while the optimization is performed over constrained nominal variables \cite{hewingRecursively2020}. In this case, the objective to minimize would be 
\begin{equation}\label{eq:nonlinear-cost}
\mathcal{J}_N = \EE{\sum_{i=0}^{N-1} \ell\left(x_{i|k}, \varphi(u_{i|k})\right)} + \EE{V_f(x_{N|k})}.
\end{equation}
For this formulation, assuming that a minimizer exists at each time step and that appropriate design constraints are imposed on the terminal weight $S$, it is possible to establish closed-loop stability and convergence properties using standard MPC arguments based on cost decrease along shifted trajectories. The properties are formalized and proven in Appendix \ref{app:proof-real-cost}. 

However, minimizing the indirect-feedback cost in the present saturation-aware framework leads to a generally nonconvex optimization problem. Indeed, the saturation nonlinearity has already been embedded into the predictive model and the probabilistic reachable set construction; reintroducing it inside the objective function would undermine the tractability gains achieved through this reformulation and render the resulting problem unsuitable for online optimization.

As such, we adopt a tractable surrogate objective defined purely in terms of the nominal state and input trajectories. Specifically, we minimize a quadratic cost of the form
\begin{equation}
    \label{eq:nom-cost}
    \mathcal{J}^n_k = \sum_{i = 0}^{N -1} \ell\left(z_{i|k}, v_{i|k}\right) + V_f(z_{N|k}) + \rho_k \xi^2_k.
\end{equation}
where $\rho_k$ is a variable positive scalar weight and $\xi_k \in [0, 1]$ for all $k$ is an interpolating variable. This cost introduces feedback with respect to the measured state $x_k$ through the initialization strategy
\begin{equation}
    \label{eq:z0-interp}
    z_{0|k} = \left(1 - \xi_k \right) x_k + \xi_k z_{1|k-1}.
\end{equation}
The variable $\xi_k$ allows a continuous interpolation between measured-state initialization and a shifted nominal trajectory, which is instrumental in preserving feasibility in the presence of unbounded disturbances. This initialization approach was originally suggested in \cite{Kohler_2022}, \cite{Schluter_2022}, \cite{Schluter_2023}. Note that other feasibility-preserving approaches may be adopted, such as that presented in~\cite{fiacchini_2025_recursive}.

\begin{proposition}[Closed-loop PRS]
\label{prop:closed-prs}
Let Assumptions \ref{ass:regularity-conds} and \ref{ass:vss-lambdastar} hold. Given the initialization rule \eqref{eq:z0-interp}, we have
\begin{equation}
\Pr\!\left\{ e_{i|k} \in \mathcal{R}^\varepsilon_{i+k}(\hat\lambda) \right\} \;\geq\; 1-\varepsilon,
\end{equation}
for all $i, \, k \in \mathbb{N}_0$.
\end{proposition}

\begin{proof}
At time $k=0$, we have $e_{0|0}=0$. By construction, the sets
$R^\varepsilon_i(\hat\lambda)$ are Markov-based PRS for $e_{i|0}$, i.e.,
$\Pr\!\left\{ e_{i|0} \in \mathcal{R}^\varepsilon_i(\hat\lambda) \right\} \ge 1-\varepsilon$ for all $i \in \mathbb{N}_0$. Assume now that $\mathcal{R}^\varepsilon_{i+k}(\hat\lambda)$ is a PRS for $e_{i|k}$ at some time $k \ge 0$. Under the initialization rule \eqref{eq:z0-interp}, the error at the next time step satisfies
$
e_{0|k+1} = x_{k+1} - z_{0|k+1} = \xi_{k+1} (x_{k+1} - z_{1|k}) = \xi_{k+1} e_{1|k}$,
which implies
$
\Ek{0|k+1} = \xi_{k+1}^2 \, \Ek{1|k} \leq \Ek{1|k},
$
since $\xi_{k} \in [0,1]$ for all $k \in \NN$. As such, the same moment bound constructed for $e_{1|k}$ holds for $e_{0|k+1}$, meaning that $\mathcal{R}^\varepsilon_{k+1}$ is a valid PRS for $e_{0|k+1}$ for all $k$ via the same Markov-based construction. Using the same recursive argument as in the proof of Proposition \ref{prop:base-exp-bound}, it follows that the shifted sets $\mathcal{R}^\varepsilon_{i+k+1}(\hat\lambda)$ are valid PRS for $e_{i|k+1}$ for all $i \geq 0$. By induction on $k$, the claim holds for all $i, \, k \in \NN_0$.
\end{proof}
\vspace{0.2cm}

The following section establishes the main closed-loop properties of the proposed SMPC scheme under this cost function and initialization strategy.

\subsection{Properties of SA-SMPC}
\label{subsec:prop-sasmpc}
Following our development thus far, the tractable SMPC formulation for the linear discrete-time system \eqref{eq:lti-sat-system} is given by the following optimization problem at each time step $k$,
\begin{subequations}
    \label{eq:sa-smpc}
    \begin{align}
        \min_{\mathbf{v}_k,\, \mathbf{z}_k,\, \xi_k} \quad \sum_{i = 0}^{N-1} \ell &\left(z_{i \mid k}, v_{i \mid k}\right) + V_f \left(z_{N|k} \right) + \rho_k \xi_k^2 \\
        \text{s. t.} \quad z_{i+1 | k} &= Az_{i | k} + Bv_{i | k}, \\
        z_{0|0} &= x_0 \label{eq:x0}\\
        z_{0|k} &= (1 - \xi_k) x_k + \xi_k \, z_{1|k-1}, && k \in \NN \label{eq:interpolation-const}, \\
        v_{i | k} &\in \mathcal{V}, &&i \in \NN^{N-1}_0,  \label{eq:final-input-const}\\
        z_{i | k} &\in \mathcal{Z}_{i+k} = \mathcal{X} \ominus \mathcal{R}^\varepsilon_{i+k}(\hat\lambda), &&i \in \NN^{N-1}_0, \label{eq:final-state-const} \\
        z_{N | k} &\in \mathcal{Z}_f \subseteq \Big( \mathcal{X} \ominus \mathcal{R}^\varepsilon (\hat\lambda) \Big) \cap \mathcal{X}_K  \label{eq:final-term-const}, \\
        \xi_k &\in [0, 1] \label{eq:interpolation-var}.
    \end{align}
\end{subequations} 
Note that constraint (\ref{eq:x0}) is only meaningful for $k = 0$, it can be dropped afterwards. Moreover, notice that the reachable sets in (\ref{eq:final-state-const}) are shifted with $k$, to maintain the feasibility (see Proposition \ref{prop:recursive-feas} below.)

\begin{assumption}[Initial Feasibility]
    \label{ass:init-feas}
    We assume that the optimization problem \eqref{eq:sa-smpc} is feasible at time $k=0$. At this initial step, the system state is exactly known, and the nominal state is initialized as $z_{0|0}=x_0$, which implies $e_0=0$.
\end{assumption}

We first establish that the proposed SMPC problem remains feasible over time, which is a prerequisite for all subsequent closed-loop guarantees. 

\begin{proposition}[Recursive Feasibility]
    \label{prop:recursive-feas}
    If Assumption \ref{ass:init-feas} is satisfied, then problem \eqref{eq:sa-smpc} is recursively feasible.
\end{proposition}
\begin{proof}
Assume that at time step $k$, an optimal input sequence $\mathbf{v}^\ast_k = \{v_{0 \mid k}^\ast, \dots, v_{N-1 \mid k}^\ast\}$ exists that satisfies all input constraints \eqref{eq:final-input-const}, state constraints \eqref{eq:final-state-const} for $i = 0, \dots, N-1$, and the terminal constraint \eqref{eq:final-term-const} for the predicted terminal state $z_{N \mid k}$. Moreover, suppose that there exists some $\xi_k$ satisfying \eqref{eq:interpolation-var} and \eqref{eq:interpolation-const}. Now, consider the subsequent time step $k+1$. A candidate solution can always be constructed by choosing $\xi_{k+1} = 1$, shifting the previous optimal sequence, and appending the terminal control law, i.e., $\mathbf{v}_{k+1} = \{v_{1 \mid k}^\ast, \dots, v_{N-1 \mid k}^\ast, K_f z_{N \mid k}\}$. The feasibility of this candidate sequence at $k+1$ follows from the optimality of $\mathbf{v}_k^\ast$, the shifting of the constraint sets as shown in \eqref{eq:final-input-const} and \eqref{eq:final-state-const}, and the positive invariance of the terminal set $\mathcal{Z}_f$ under the terminal control law. Specifically, the first $N-1$ elements of $\mathbf{v}_{k+1}$ were already part of the feasible solution at time $k$ and therefore trivially satisfy the input and state constraints at time $k+1$.  Furthermore, since $z_{N \mid k} \in \mathcal{Z}_f$ and the set $\mathcal{Z}_f$ is positively invariant under the control law $u = K_f z$, it follows that the successor state remains in $\mathcal{Z}_f$ while the associated input $K_f z_{N \mid k}$ satisfies the input constraints for all future time steps. Therefore, $\mathbf{v}_{k+1}$ is a suboptimal but feasible solution at time $k+1$ for all $k \in \NN$.
\end{proof}

Having established recursive feasibility, we now show that the proposed tightening strategy ensures satisfaction of the original chance constraints.
\begin{proposition}[Chance Constraint Satisfaction]
    Given Assumptions \ref{ass:vss-lambdastar} and \ref{ass:init-feas}, the closed-loop PRS and recursive feasibility properties of Propositions \ref{prop:closed-prs} and \ref{prop:recursive-feas}, the proposed SA-SMPC scheme guarantees satisfaction of the chance constraints \eqref{eq:xik-constr} and \eqref{eq:xn-constr} conditioned on the known initial state $x_0$.
\end{proposition}

\begin{proof}
At time $k=0$, the initial state is exactly known. Proposition \ref{prop:closed-prs} guarantees that the error process satisfies
\begin{equation*}
\Pr\left\{ e_{i|k} \in \mathcal{R}^\varepsilon_{i+k}(\hat\lambda) \, \middle| \, x_0 \right\} \geq 1-\varepsilon,
\qquad \forall i,k \in \mathbb{N}.
\end{equation*}
From Proposition \ref{prop:recursive-feas}, constraint $z_{i|k} \in \mathcal{X} \ominus \mathcal{R}^\varepsilon_{i+k}(\hat\lambda)$, which implies that $\Pr\left\{x_{i|k} \in \mathcal{X} \mid x_0\right\} \geq 1 - \varepsilon$, is satisfied for every $i, k \in \NN$.
\end{proof}

We next analyze the closed-loop performance induced by the nominal cost and initialization strategy.
\begin{proposition}[Nominal Cost Bound]
\label{prop:cost-decrease-bound}
Let Assumptions \ref{ass:vss-lambdastar} and \ref{ass:init-feas} hold. For problem \eqref{eq:sa-smpc} defined above, the optimal cost function $\mathcal J_k^\ast$ satisfies
\begin{equation}
    \label{eq:1step-cost-decrease}
    \mathcal J^\ast_{k+1} - \mathcal J^\ast_k \leq - \ell(z^\ast_{0|k}, v^\ast_{0|k}) + \rho_{k+1}.
\end{equation}
for all $k \geq 0$. Consequently, for any integer $\bar{N} \geq 1$,
\begin{equation}
\label{eq:cost-bound}
\sum_{k=0}^{\bar N-1} \ell(z^\ast_{0|k}, v^\ast_{0|k}) \leq \mathcal J^\ast_0 + \sum_{k=1}^{\bar{N}} \rho_k .
\end{equation}
\end{proposition}
\vspace{0.5cm}
\begin{proof}
At time $k+1$, consider the shifted candidate sequences and interpolating variable
\begin{equation*}
z_{i|k+1} = z^\ast_{i+1|k}, \qquad
v_{i|k+1} = v^\ast_{i+1|k}, \qquad
\xi_{k+1} = 1,
\end{equation*}
with terminal input $v_{N-1|k+1} = K_f z^\ast_{N|k}$ and state $z_{N|k+1} = (A+B K_f) z^\ast_{N|k}$. These candidate sequences are feasible. The associated cost satisfies
\begin{equation*}
\tilde{\mathcal J}_{k+1} \leq \mathcal{J}^\ast_k - \ell(z^\ast_{0|k}, v^\ast_{0|k}) - \rho_k \xi_k^2 + \rho_{k+1},
\end{equation*}
where the inequality
\begin{equation*}
\ell(z_{N-1|k+1}, v_{N-1|k+1}) + V_f(z_{N|k+1}) = \|z^\ast_{N|k}\|_Q^2 + \|K_f z^\ast_{N|k}\|_R^2 + \|(A+BK_f)z^\ast_{N|k}\|_S^2 \leq \|z^\ast_{N|k}\|_S^2,
\end{equation*}
which follows from \eqref{eq:terminal-lyap}, has been used. Since $\mathcal{J}^\ast_{k+1} \leq \tilde{\mathcal{J}}_{k+1}$ and $\rho_k \xi_k^2 \ge 0$, the cost decrease relationship in \eqref{eq:1step-cost-decrease} follows. Summing the result from $k=0$ to $\bar N-1$ and using $\mathcal{J}^\ast_{\bar N} \geq 0$
yields the cost bound in \eqref{eq:cost-bound}.
\end{proof}

The following result shows that the nominal cost decrease implies convergence of the nominal state and input.
\begin{proposition}[Nominal State and Input Convergence]
\label{prop:nom-state-conv}
Let Assumptions \ref{ass:vss-lambdastar} and \ref{ass:init-feas} hold. If $\sum_{k=1}^\infty \rho_k < \infty$, then
$\lim_{k \to \infty} z^\ast_{0|k} = 0$ and $\lim_{k \to \infty} v^\ast_{0|k} = 0$.
\end{proposition}

\begin{proof}
From Proposition \ref{prop:cost-decrease-bound}, and $\sum_{k=1}^\infty \rho_k < \infty$, the right-hand side of (\ref{eq:cost-bound}) is uniformly bounded, implying $\sum_{k=0}^\infty \ell(z^\ast_{0|k}, v^\ast_{0|k}) < \infty$.
Since $\ell(\cdot,\cdot)$ is positive definite, the result holds.
\end{proof}

Finally, we relate the convergence of the nominal system to mean-square stability of the true closed-loop state.
\begin{proposition}[Mean-square Stability of $x_k$]
    Let Assumptions \ref{ass:vss-lambdastar} and \ref{ass:init-feas} hold. Given $x_k = z_k + e_k$, and $z_k \to 0$, the true state $x_k$ is mean-square stable, satisfying
\begin{equation}
    \lim_{k\to\infty} \EE{\norm{x_k}{Q}} \leq \Lambda_{\max}\left(P^{-1/2} Q P^{-1/2}\right) \frac{1}{1-\hat\lambda}\Tr(PW),
\end{equation}
\end{proposition}

\begin{proof} From Proposition \ref{prop:tighter-bound-on-exp} and $Q \preceq \Lambda_{\max}\left(P^{-1/2} Q P^{-1/2}\right) P$, we can derive the following bound
\begin{equation*}
\label{eq:scaled-exp-bound}
\EE{\norm{e_k}{Q}} \leq \Lambda_{\max}\left(P^{-1/2} Q P^{-1/2}\right) \frac{1- (\hat\lambda)^k}{1-\hat\lambda}\Tr(PW).
\end{equation*}
Using $x_k = z_k + e_k$, we have $\norm{x_k}{Q} = \norm{z_k}{Q} + \norm{e_k}{Q} + 2 z_k^\top Q e_k$. Taking expectations and applying the Cauchy-Schwarz inequality yields
\begin{equation*}
 \EE{\norm{x_k}{Q}} \leq \norm{z_k}{Q} + \EE{\norm{e_k}{Q}} + 2 \sqrt{\norm{z_k}{Q} \EE{\norm{e_k}{Q}}}.
\end{equation*}
Hence,
\begin{equation*}
\EE{\norm{x_k}{Q}} \leq \left(\sqrt{\norm{z_k}{Q}} + \sqrt{\EE{\norm{e_k}{Q}}} \right)^2 .
\end{equation*}
Using $z_k \to 0$ from Proposition \ref{prop:nom-state-conv} and taking the limit as $k \to \infty$ gives the stated result. 
\end{proof}

The algorithm resulting from the optimization problem in \eqref{eq:sa-smpc} and the preceding development is presented below.
\begin{algorithm}[H]
\caption{Offline Design and Online Execution}
\label{alg:sa-smpc}
\begin{algorithmic}[1]
\REQUIRE System matrices $A$, $B$, constraints $\mathcal{X}$, $\mathcal{U}$, noise covariance $W$, cost matrices $Q,\,R$, violation level $\varepsilon$, parameter $v_\mathrm{ss}$.

\textbf{Offline Design}
\STATE Solve \eqref{eq:opt-params} for contraction rates $\lambda, \lambda_L$, Lyapunov matrix $P$, and feedback gain $K$.
\STATE Determine effective contraction rate $\hat\lambda$ corresponding to $v_\mathrm{ss}$.
\STATE Compute tightened state constraints $\mathcal{Z}_i$, and terminal set $\mathcal{Z}_f$.

\textbf{Online Execution}
\STATE Initialize $z_{0|0} = z_0  \gets x_0$ (measured-state initialization).
\FOR{$k = 0, 1, \dots$}
    \STATE Measure real state $x_k$.
    \STATE Solve \eqref{eq:sa-smpc} for variable $\xi_k^*$ and corresponding sequences $\mathbf{z}^*_k$, $\mathbf{v}^*_k$.
    \STATE Apply $u_k \gets \varphi\left(v^*_{0|k} + K(x_k - z^\ast_{0|k})\right)$.
    \STATE Store $z^\ast_{1|k}$ for next iteration. 
\ENDFOR
\end{algorithmic}
\end{algorithm}

Enforcing \eqref{eq:final-input-const} ensures the nominal input $v_{i \mid k}$ is admissible, but cannot guarantee that the total control input $u_k = v_{0 \mid k}^\ast + Ke_k$ satisfies the original input constraints \eqref{eq:u-constr}. However, by accounting for the saturation nonlinearity into the predictive model, the constructed PRS reflect the true closed-loop behavior of the system, thus enabling the controller to anticipate the potential saturation and plan against its effects. This means that the derived closed-loop properties are valid under the physically implementable control law $\varphi\left(u_k\right)$. In contrast, standard constraint-tightening SMPC formulations such as \cite{hewingStochastic2018}, establish guarantees under the idealized linear control law $u_k = v_k + Ke_k$, which becomes physically unrealizable when saturation occurs, thereby invalidating theoretical assurances of state constraint satisfaction and stability. Our proposed method replaces this weak assumption with a strong, verifiable guarantee of robust state safety and stability under real operating conditions.

\begin{remark}[Input Constraints Handling]
    Note that $r_L$ depends, in general, on the nominal input $v$, see \eqref{eq:rL},which directly influences $\bar\lambda^\ast$. This introduces a circular dependency: large nominal inputs shrink the region of linearity and lead to conservative PRS, yet computing tight PRS requires prior knowledge of $v$, which itself depends on the resulting tightened constraints. A naïve solution is to rely solely on the open-loop induced PRS $\mathcal{R}^\varepsilon_i(\lambda)$ for the constraint tightening. While tractable and independent of $v$, this approach may be  highly conservative and cause infeasibility. Similarly, although the strategy employed, by imposing Assumption \ref{ass:vss-lambdastar}, supports the theoretical development by ensuring the existence of an effective contraction rate, it may become overly restrictive in practice. In particular, for moderate or large disturbance levels, enforcing a constant input bound $v_\mathrm{ss}$ along the prediction horizon can excessively limit the nominal input and render the nominal trajectory infeasible. An alternative strategy would consist in adopting a structured heuristic in which the nominal input is parameterized offline as a monotonically decreasing sequence $\bar{v}_i$, ranging from near-saturation toward the steady-state bound $v_\mathrm{ss}$. The decay rate may be tuned according to design preferences. For instance, following a robust-style argument, $\bar{v}_i$ can be chosen to decrease sufficiently fast to reserve adequate control authority for the feedback action associated with the worst probabilistic error evolution characterized by $\mathcal{R}^\varepsilon_i (\lambda)$, while maintaining $\bar{v}_i \geq v_\mathrm{ss}$. From a set-theoretic perspective, this corresponds to tightening the input constraints by bounded scalings $K \mathcal{R}^\varepsilon_i(\lambda)$. Related ideas appear in the literature; for instance, \cite{paulsonStochastic2020} enforces a saturation on the disturbance to preserve a non-empty set of admissible nominal inputs.
\end{remark}

\section{Numerical Example}
\label{sec:numerical}
This section illustrates the practical applicability of the proposed SA-SMPC framework on a representative system. All simulations were carried out in Python using the \emph{CVXPY} modeling interface \cite{diamond2016cvxpy, agrawal2018rewriting} in conjunction with the \emph{CLARABEL} solver \cite{Clarabel_2024}. The computations were executed on a Linux workstation equipped with an AMD Ryzen 7700 CPU (16 cores) and 32GB of RAM.

For comparison purposes, we evaluate the proposed method with  $\mathcal{R}^\varepsilon_i (\bar\lambda^\ast)$-based constraint tightening against two alternatives: the conservative tightening using $\mathcal{R}^\varepsilon_i (\lambda)$, and the affine disturbance-feedback parameterization approach proposed in \cite{paulsonStochastic2020} with fixed uniform risk allocation. The latter is particularly relevant, as it represents an alternative analytical SMPC framework capable of enforcing hard input constraints under unbounded disturbances. Since it results in a second-order cone program (SOCP) solved online at each time step, we refer to it as SOCP-SMPC. The three approaches will thus be denoted as $\bar\lambda^\ast$-SA-SMPC, $lambda$-SA-SMPC, and SOCP-SMPC, respectively.

\begin{remark}[Uniform Risk Allocation]
\label{rem:fixed-risk}
    It is worth noting that \cite{paulsonStochastic2020} also discusses the possibility of optimizing the risk allocation across constraints. However, the joint optimization of the feedback gain and the risk allocation is inherently nonconvex. Although convex approximations have been proposed, they substantially increase computational effort and do not guarantee global optimality. The iterative procedure suggested in \cite{paulsonStochastic2020} alternates between optimizing the risk allocation for a fixed feedback gain and optimizing the feedback gain for a fixed risk allocation, repeating this process for a prescribed number of iterations. In practice, each iteration effectively entails solving an additional SOCP of comparable complexity, leading to a significant increase in computational burden. For this reason, we adopt a fixed uniform risk allocation in the SOCP-SMPC implementation.
\end{remark}

\subsection{Setup}
We consider an isothermal continuous-stirred-tank reactor (CSTR) as in \cite{HeirungSMPC2018}, where the states $x_1$ and $x_2$ denote the concentrations of species $A$ and $B$, respectively, and the control input $u$ is the dilution rate. The control objective is to regulate $x_2$ to the setpoint $x_2^\ast=1$, which corresponds to the steady-state operating point $(x_1^\ast,\,u^\ast)=(2.5,\,25)$. Linearizing the nonlinear CSTR dynamics about this equilibrium and discretizing with sampling time $T_s=0.002\,\mathrm{s}$ yields the discrete-time deviation model
$$ x_{k+1} = Ax_k + Bu_k + w_k, $$
with
$$
A=\begin{bmatrix} 0.95123 & 0 \\ 0.08833 & 0.81873 \end{bmatrix},
\qquad
B=\begin{bmatrix} -0.0048771 \\ -0.0020429 \end{bmatrix}.
$$
This model coincides with \cite{HeirungSMPC2018}; in the present numerical study, the cost function weights, disturbance statistics, state constraints and initial conditions are adapted as specified below. Process disturbances are modeled as i.i.d.\ Gaussian noise $w_k \sim \mathcal{N}(0,W)$ with $ W = 0.003^2 I_2$. We fix $Q = \mathrm{diag} (20, 100)$ and $R = 0.1$, and enforce hard input bounds  $-25 \leq \tilde{u}_k \leq 25$. We will treat 4 different variations of the problem with changing state constraints and initial conditions (both specified in deviation variables). The considered variations are intentionally selected to challenge both feasibility and performance margins of the different control strategies. Specifically, \begin{itemize}
    \item Scenario 1: $\Pr\{\tilde{x}_{2} \leq 0.25 \} \geq 0.8$, $\tilde{x}_0 = [0.5,\; 0.2]^\top$,
    \item Scenario 2: $\Pr\{\tilde{x} \leq [0.75, \; 0.25]^\top \} \geq 0.8$, $\Pr\{2\tilde{x}_1 + \tilde{x}_2 \leq 1.5\} \geq 0.8$, $\tilde{x}_0 = [0.5,\; 0.2]^\top$,
    \item Scenario 3: $\Pr\{[-0.5, \; -0.25]^\top\!\leq \tilde{x} \leq [0.75, \; 0.25]^\top \} \geq 0.8$, $\Pr\{2\tilde{x}_1 + \tilde{x}_2 \leq 1.5\} \geq 0.8$, $\tilde{x}_0 = [0.5,\; 0.2]^\top$,
    \item Scenario 4: $\Pr\{\tilde{x}_2 \leq 0.15\} \geq 0.8$, $\tilde{x}_0 = [0.5,\; 0.12]^\top$.
\end{itemize}
As such, we set a $\varepsilon = 0.2$ for the SA-SMPC methods, and a risk allocation $\varepsilon_i = 0.2 / |\NN_1^{\mathrm{nc}}|$, where $\mathrm{nc}$ denotes the number of constraints, for the SOCP-SMPC strategy.
Note that while the constraints on the combined states introduced in variations 2 and 3 may not always correspond to physically meaningful scenarios, they provide a useful benchmark for assessing the feasibility limits of any given SMPC approach. All scenarios are run with a prediction horizon $N = 15$.

\subsection{Offline Design for SA-SMPC}

For all variations of the problem, the offline design process of the SA-SMPC scheme is identical, and yields the same results, with only the terminal set $\mathcal{Z}_f$ changing due to the addition of constraints. Solving the offline optimization problem in \eqref{eq:opt-params}, we obtain an open-loop contraction $\lambda \approx 0.9049$, a linear contraction rate $\lambda_L \approx 0.67035$, along with the feedback gain $K$ and matrix $P$ as specified below
$$
K = \begin{bmatrix} 27.1573 & 0.09622 \end{bmatrix}, \qquad P = \begin{bmatrix}
    71.2230 & -0.4498 \\
    -0.4498 & 1.01712 
\end{bmatrix}.
$$

Given that the noise magnitude is relatively low, we set $v_{i|k} \leq v_\mathrm{ss} = 24$ for all $i,\, k$ following Assumption \ref{ass:vss-lambdastar}, for which we compute an effective contraction rate $\bar\lambda^\ast \approx 0.6752$. The constraint tightening that follows is standard, as described in Section \ref{sec:sa-smpc}.

\subsection{Comparison and Discussion}
\begin{table}[t]
  \centering
  \caption{Comparison of $\bar\lambda^\ast$-SA-SMPC, $\lambda$-SA-SMPC, and SOCP-SMPC methods on the basis of average cost and computational time for each of the four studied scenarios.}
  \label{tab:mpc_comparison}
  \setlength{\tabcolsep}{6pt}
  \renewcommand{\arraystretch}{1.15}
  \begin{tabular}{
    c
    >{} S
    >{} S
    >{} S
    >{} S
    >{} S
    >{} S
    >{} S
  }
    \toprule
    & \multicolumn{2}{c}{\textbf{Novel $\boldsymbol{\bar\lambda^\ast}$-SA-SMPC}}
    & \multicolumn{2}{c}{\textbf{$\boldsymbol\lambda$-SA-SMPC}}
    & \multicolumn{2}{c}{\textbf{SOCP-SMPC}} \\
    \cmidrule(lr){2-3}\cmidrule(lr){4-5}\cmidrule(lr){6-7}
    \textbf{Scenario}
    & {$\E_w\left[\mathcal{J}_\mathrm{MPC}\right] \pm \sigma_\mathrm{MPC}$} & {$\bar\tau$ (ms/iter)}
    & {$\E_w\left[\mathcal{J}_\mathrm{MPC}\right] \pm \sigma_\mathrm{MPC}$} & {$\bar\tau$ (ms/iter)}
    & {$\E_w\left[\mathcal{J}_\mathrm{MPC}\right] \pm \sigma_\mathrm{MPC}$} & {$\bar\tau$ (ms/iter)} \\
    \midrule
    Scenario 1 & {$77.1 \pm 1.95$} & {1.62} & {79.2 $\pm 2.11$} & {1.65} & {67.0 $\pm 1.84$} & {52} \\
    Scenario 2 & {$77.1 \pm 1.95$} & {1.57} & {79.2 $\pm 2.11$} & {1.63} & {88.8 $\pm 1.55$} & {146} \\
    Scenario 3 & {$77.1 \pm 1.95$} & {1.64} & {79.2 $\pm 2.11$} & {1.71} & {Infeasible} & {Infeasible}\\
    Scenario 4 & {$114.0 \pm 2.27$} & {1.59} & {Infeasible} & {Infeasible} & {$57.5 \pm 1.77$} & {57} \\
    \bottomrule
  \end{tabular}
\end{table}

The strategies are compared over 1000 trajectories on the basis of feasibility, average computational time $\bar\tau$ per iteration, and performance in terms of the index $\mathcal{J}_\mathrm{MPC} = \sum_{k = 0}^{T} \left( \norm{\tilde{x}_k}{Q} + \norm{\tilde{u}_k}{R} \right)$. The main results are summarized in table \ref{tab:mpc_comparison}.

For the two variations in which all strategies are feasible, several trends emerge. In scenario 1, the SOCP-SMPC scheme exhibits superior performance, primarily due to its ability to optimize the feedback gain online and thereby accurately capture the effect of low noise magnitudes. This performance gain, however, comes at a substantial computational expense, as the SOCP-SMPC formulation is approximately 35 times more demanding than either SA-SMPC scheme. We also note that the reduced conservatism introduced by the refined tightening enables $\bar\lambda^\ast$-SA-SMPC to outperform its open-loop counterpart.

Scenario 2 augments the constraint set of scenario 1 with two additional constraints, leading to a significant degradation in the performance of SOCP-SMPC. This deterioration stems from the treatment of constraints as disjoint events via Boole’s inequality (see \cite{paulsonStochastic2020}), which effectively assigns each of the three constraints a violation probability of $0.2 / 3 \approx 0.667$, thereby inducing substantially more conservative behavior. As mentioned in Remark~\ref{rem:fixed-risk}, \cite{paulsonStochastic2020} proposes an online iterative procedure for risk allocation optimization per individual constraint that may mitigate this effect, such an approach further increases the computational effort and is unlikely to remain effective for higher-dimensional systems. Consistent with this observation, the computational time of SOCP-SMPC increases nearly three-fold to accommodate the additional constraints, and thus becomes around 100 times more computationally demanding than any of the SA-SMPC approaches. In contrast, both SA-SMPC methods display identical performance to that observed in Scenario 1, outperforming the SOCP-SMPC scheme, with no change in computational effort. The performance deterioration trend for the SOCP-SMPC scheme devolves into infeasibility in Scenario 3, where the addition of more constraints render the risk allocation too tight, as reported in table \ref{tab:mpc_comparison}.

In scenario 4, the $\lambda$-SA-SMPC scheme becomes infeasible due to the tighter constraint on $x_2$. Once again, SOCP-SMPC achieves superior performance. By comparison, the conservatism of the $\bar\lambda^\ast$-SA-SMPC formulation forces the mean trajectory to remain excessively backed off from the constraints, thereby requiring a greater dilution rate which increases the cost function value. This can be seen more clearly in figure~\ref{fig:scenario-4-trajs} displaying the resulting input profiles of the two strategies.

It is worth noting that, although not explicitly illustrated in the presented results, the SOCP-SMPC scheme is expected to remain feasible under higher noise magnitudes and for initial conditions arbitrarily close to the constraint boundaries, whereas the SA-SMPC schemes lose feasibility in such scenarios.

Overall, our proposed SA-SMPC framework is computationally attractive, exhibiting low complexity and favorable scalability with respect to both system dimension and the number of constraints. The effective tightening of $\bar\lambda^\ast$-SA-SMPC further reduces conservatism, resulting in improved performance and enlarged feasibility regions relative to $\lambda$-SA-SMPC. However, this computational efficiency is achieved at the expense of increased conservatism, in no small part due to the use of the loose Markov inequality and fixed reachable set geometries, which leads to overly cautious constraint back-offs and reduced performance and feasibility in certain scenarios. The principal advantage of SOCP-SMPC lies in its nonconservative back-off strategy, which supports reduced costs, larger feasibility regions under tighter constraints and elevated noise levels, and satisfactory constraint violation frequencies. Nevertheless, this benefit is offset by poor computational scalability and increased sensitivity to the number of constraints, suggesting limited suitability for high-dimensional applications.

\begin{figure*}[t]
  \centering
  \begin{minipage}[b]{0.48\linewidth}
    \centering
          \begin{overpic}[width=\textwidth]{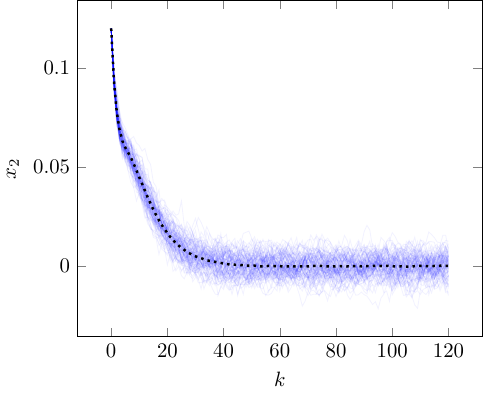}
        \put(45,35){%
          \includegraphics[width=0.5\textwidth]{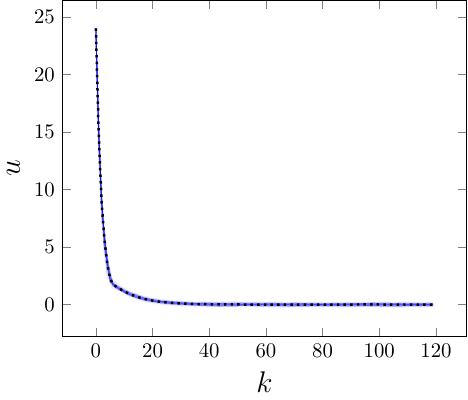}%
        }
      \end{overpic}
    \par\small (a)
  \end{minipage}\hfill
  \begin{minipage}[b]{0.48\linewidth}
    \centering
    \begin{overpic}[width=\textwidth]{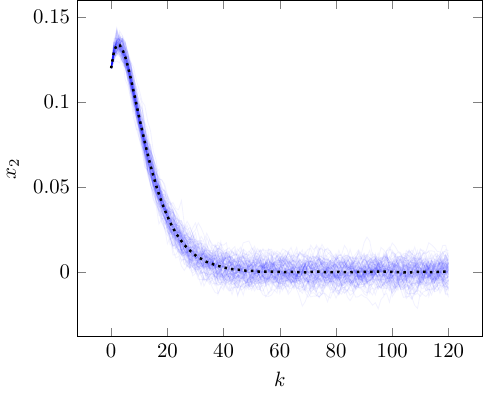}
        \put(45,35){%
          \includegraphics[width=0.5\textwidth]{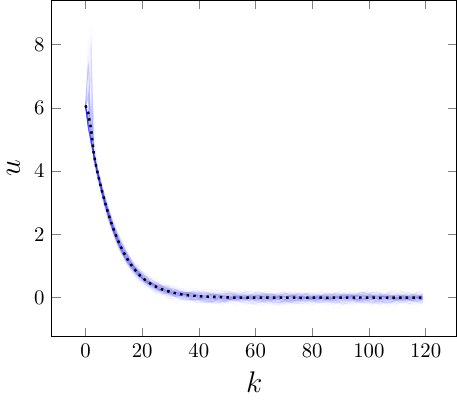}%
        }
      \end{overpic}
    \par\small (b)
  \end{minipage}
  \caption{Trajectory of state $x_2$ in scenario 4 under (a) $\bar\lambda^\ast$-SA-SMPC and (b) SOCP-SMPC. The inset in (a) shows the corresponding input profile for the $\bar\lambda^\ast$-SA-SMPC strategy, which exploits the full admissible input range. The inset in (b) reports the input profile for SOCP-SMPC, which operates over a restricted subset of the admissible inputs.}
  \label{fig:scenario-4-trajs}
\end{figure*}

\section{Conclusion}

We presented a tractable stochastic tube-based SMPC framework that guarantees hard input-constraint satisfaction for linear systems under unbounded disturbances. By incorporating actuator saturation directly into the prediction model and constructing PRS via sector bounds and convex embeddings, we obtain a control scheme that preserves a linear nominal optimization while appropriately characterizing the resulting nonlinear error dynamics. This saturation-based framework provides a systematic probabilistic link between linear and saturated regimes, enabling tighter PRS and improved performance. As demonstrated in the numerical study, the proposed approach outperforms schemes based on naïve open-loop tightenings, while remaining significantly more computationally efficient than its feedback-parameterized counterpart, albeit at the cost of increased conservatism. Future work aims to further reduce said conservatism by using the moment sum-of-squares hierarchy and polynomial optimization tools to incorporate higher-order disturbance moments into the PRS design, enabling the construction of semialgebraic PRS certified by moment-based probabilistic certificates that are less conservative than Markov-type bounds.

\appendices
\section{Proof of Proposition 3} \label{app:prop3-proof}

Given $\lambda - \lambda_L \geq 0$, inequality \eqref{eq:prob-exp-bound} implies that the globally valid upper bound is
\begin{equation*}
   \Ek{i+1|k} \leq \lambda \Ek{i|k} + \Tr(PW),
\end{equation*}
which holds when $\Pr\left\{e_{i|k} \in \mathcal{S}_1\right\} = 0$. Iterating this recursion yields
\begin{equation*}
   \begin{split}
     \Ek{i|k} &\leq \frac{1 - \lambda^i}{1 - \lambda} \Tr(PW) \leq \frac{1}{1 - \lambda} \Tr(PW),
   \end{split}
\end{equation*}
for all $i, \, k \in \NN$ and every $e_{i|k} \in \R^n$. Condition (\ref{eq:rl}) guarantees the existence of a parameter $\mu \in \left[\lambda_L,\; \lambda \right]$ such that
\begin{equation}\label{eq:exp/rl-param}
    \frac{\Ek{i|k}}{r_L} \leq \frac{\mu - \lambda_L}{\lambda - \lambda_L} \leq 1.
\end{equation}
Substituting this parametrized bound into (\ref{eq:exp-bound-tight}) produces
\begin{equation}\label{eq:exp-bound-bar}
   \Ek{i+1|k} \leq \mu \Ek{i|k} + \Tr(PW),
\end{equation}
which mirrors (\ref{eq:recursive-exp-bound}) with parameter $\mu$. Iteration then gives
\begin{equation}\label{eq:bound-bar}
      \Ek{i|k} \leq \frac{1 - \mu^i}{1 - \mu} \Tr(PW) \leq \frac{1}{1 - \mu} \Tr(PW),
\end{equation}
provided (\ref{eq:exp/rl-param}) holds for all $i, \, k \in \NN$, thereby ensuring (\ref{eq:bound-bar}). 
We now identify the values of $\mu$ for which (\ref{eq:exp/rl-param}) holds universally, thus validating (\ref{eq:bound-bar}). Define the function
\begin{equation}
   g(\mu) \coloneqq \frac{\mu - \lambda_L}{\lambda - \lambda_L} r_L - \frac{1}{1 - \mu} \Tr(PW),
\end{equation}
which is continuous and strictly concave on the interval $\left[ \lambda_L, \lambda \right]$. From (\ref{eq:rl}), $g(\lambda) > 0$. At $\mu = \bar\lambda^\ast$, defined by (\ref{eq:lambdastar}), $g(\bar{\lambda}^\ast) = 0$. The strict concavity of $g$ implies $g(\mu) > 0$ for all $\mu \in (\bar{\lambda}^\ast, \lambda]$. Equivalently,
\begin{equation}\label{eq:mu}
\frac{1}{1 - \mu} \Tr(PW) < \frac{\mu - \lambda_L}{\lambda - \lambda_L} r_L, \quad \forall \mu \in (\bar{\lambda}^\ast, \lambda].
\end{equation}
Now, fix $\mu \in (\bar{\lambda}^\ast, \lambda]$. Three potential cases arise for $\Ek{i|k}$:
\begin{enumerate}
   \item $\displaystyle{\Ek{i|k} \leq \frac{1}{1 - \mu} \Tr(PW),}$ \label{case1}
   \item $\displaystyle{\frac{1}{1 - \mu} \Tr(PW) < \Ek{i|k} \leq \frac{\mu - \lambda_L}{\lambda - \lambda_L} r_L,}$ \label{case2}
   \item $\displaystyle{\frac{\mu - \lambda_L}{\lambda - \lambda_L} r_L < \Ek{i|k}}$. \label{case3}
\end{enumerate}
Case \ref{case2} leads to a contradiction, as it violates the implication that (\ref{eq:exp/rl-param}) leads to (\ref{eq:bound-bar}).
Case \ref{case3} is also impossible: by the concavity of $g$, if $\Ek{i|k}$ exceeded $\frac{\mu - \lambda_L}{\lambda - \lambda_L} r_L$, there would exist a $\nu \in (\mu,\lambda]$ such that
\begin{equation*}
   \frac{1}{1-\nu}\Tr(PW) < \Ek{i|k} = \frac{\nu - \lambda_L}{\lambda - \lambda_L}r_L,
\end{equation*}
where the left-hand inequality follows from (\ref{eq:mu}). This is equivalent to Case \ref{case2}; a contradiction.\\
Therefore, only Case \ref{case1} is possible. Similar reasoning applies for $\mu = \bar{\lambda}^\ast$.
This demonstrates that any $\mu \in [\bar\lambda^\ast,\lambda]$ ensures (\ref{eq:exp-bound-bar}) and (\ref{eq:bound-bar}) for all $i, \, k \in \NN$. The smallest such $\mu$, i.e., that for which $g(\mu) = 0$, is $\mu = \bar\lambda^\ast$ and yields the tightest valid upper bound
\[
   \Ek{i|k} \leq \frac{1 - (\bar\lambda^\ast)^i}{1 - \bar\lambda^\ast} \Tr(PW).
\]  
This establishes (\ref{eq:exp-bound-bar-cf}) and concludes the proof. \endproof

\section{Convergence and Mean-square Stability under the Real Cost}
\label{app:proof-real-cost}
\begin{proposition}[Average Cost Bound.]
\label{prop:real-cost-bound}
Let Assumptions \ref{ass:regularity-conds} and \ref{ass:vss-lambdastar} hold, and let the MPC cost function be defined as in \eqref{eq:nonlinear-cost}. For $\delta > 0$ and terminal weight $S =\alpha P$ with $\alpha \geq 1$ and $\delta \leq (1 - \hat\lambda) / \hat\lambda$ such that
\begin{equation}
\label{eq:P-scaling}
\alpha \left(1 - (1+\delta) \cdot \hat\lambda \right) P \succeq Q + K^\top R K,
\end{equation}
the optimal cost function $\mathcal J_k^\ast$ satisfies
\begin{equation}
\EE{\mathcal{J}^\ast_{k+1}} - \EE{\mathcal{J}^\ast_k} \leq -\ell\!\left(x^\ast_{0|k},\varphi\!\left(u^\ast_{0|k}\right)\right) + \left(1+\frac{1}{\delta}\right) \frac{4 }{1 - \hat\lambda} \Tr(SW),
\end{equation}
for all $k \geq 0$, and consequently
\begin{equation}
\label{eq:nonlinear-cost-bound}
\lim_{\bar{N}\to\infty}\frac{1}{\bar{N}}\sum_{k=0}^{\bar{N}-1}
\EE{\ell\!\left(x^\ast_{0|k},\varphi(u^\ast_{0|k})\right)}
\leq \left(1+\frac{1}{\delta}\right) \frac{4 }{1 - \hat\lambda} \Tr(SW).
\end{equation}
\end{proposition}

\begin{proof}
Consider the shifted candidate, whose cost is denoted by $\tilde{\mathcal{J}}_{k+1}$, obtained by applying the standard shift, and defined as:
$\tilde{x}_{i|k+1} = x^\ast_{i+1|k}$ for $i \in \NN_0^{N-1}$;
$\tilde{u}_{i|k+1} = u^\ast_{i+1|k}$ for $i \in \NN_0^{N-2}$;
$\tilde{x}_{N|k+1} = A_K z^\ast_{N|k} + e_{N+1|k}$ and
$\tilde{u}_{N-1|k+1} = K x^\ast_{N|k}$, where the dependence on the realization of $w_k$ is left implicit. Then
\begin{equation*}
\begin{split}
\EE{\tilde{\mathcal{J}}_{k+1}}
&= \EE{\sum_{i=0}^{N-1}\norm{\tilde{x}_{i|k+1}}{Q} + \norm{\sat{\tilde{u}_{i|k+1}}}{R}} + \EE{\norm{\tilde{x}_{N|k+1}}{S}} \\
&= \EE{\sum_{i=1}^{N-1}\norm{x^\ast_{i|k}}{Q} + \norm{\sat{u^\ast_{i|k}}}{R}}
+ \EE{\norm{x^\ast_{N|k}}{Q} + \norm{\sat{K x^\ast_{N|k}}}{R}}
+ \EE{\norm{A_K z^\ast_{N|k} + e_{N+1|k}}{S}} .
\end{split}
\end{equation*}
The optimal cost at time $k$ satisfies
\begin{equation*}
\begin{split}
\EE{\mathcal{J}^\ast_{k}}
&= \EE{\sum_{i=0}^{N-1}\norm{x^\ast_{i|k}}{Q} + \norm{\sat{u^\ast_{i|k}}}{R}} + \EE{\norm{x^\ast_{N|k}}{S}} \\
&= \norm{x^\ast_{0|k}}{Q} + \norm{\sat{u^\ast_{0|k}}}{R}
+ \EE{\sum_{i=1}^{N-1}\norm{x^\ast_{i|k}}{Q} + \norm{\sat{u^\ast_{i|k}}}{R}}
+ \EE{\norm{x^\ast_{N|k}}{S}} ,
\end{split}
\end{equation*}
Subtracting the two expressions yields the cost difference
\begin{equation*}
\begin{split}
\EE{\tilde{\mathcal{J}}_{k+1}} \hspace{-0.1cm}  - \EE{\mathcal{J}^\ast_{k}}
&= -\norm{x^\ast_{0|k}}{Q} - \norm{\sat{u^\ast_{0|k}}}{R} - \EE{\norm{x^\ast_{N|k}}{S}} + \EE{\norm{x^\ast_{N|k}}{Q} + \norm{\sat{K x^\ast_{N|k}}}{R}} \\ 
&\hspace{15pt} + \EE{\norm{A_K z^\ast_{N|k} + e_{N+1|k}}{S}}\\
& \leq -\norm{x^\ast_{0|k}}{Q} \hspace{-0.2cm} - \norm{\sat{u^\ast_{0|k}}}{R}  - \EE{\norm{x^\ast_{N|k}}{S}} + \EE{\norm{x^\ast_{N|k}}{Q} + \norm{K x^\ast_{N|k}}{R}} \\ 
&\hspace{15pt} + \EE{\norm{A_K x^\ast_{N|k} + e_{N+1|k} - A_K e_{N|k}}{S}}, \\ 
&\leq -\norm{x^\ast_{0|k}}{Q} - \norm{\sat{u^\ast_{0|k}}}{R} + \EE{(x^\ast_{N|k})^\top \Big( Q + K^\top R K + (1 + \delta) A_K^\top S A_K - S \Big) (x^\ast_{N|k})}, \\ 
&\hspace{15pt} + \left(1 + \frac{1}{\delta}\right) \EE{\norm{e_{N+1|k} - A_K e_{N|k}}{S}},\\
\end{split}
\end{equation*}
where the first inequality follows from the fact that $\norm{\sat{u}}{R} \leq \norm{u}{R}$ for all $u$ and the fact that $z^\ast_{N|k} = x^\ast_{N|k} - e_{N|k}$. The second inequality is obtained by applying Young's inequality, i.e. $2xy \leq \delta x^2 + (1/\delta) y^2$ for all $x, y \in \R$, in fact
\begin{equation*}
\norm{a + b}{S} \leq \norm{a}{S} + \norm{b}{S} + 2 \, \|a\|_S \, \|b\|_S \leq (1 + \delta) \norm{a}{S} + \left(1 + \frac{1}{\delta}\right) \norm{b}{S},
\end{equation*}
with $\delta > 0$, to the term $\EE{\norm{A_K x^\ast_{N|k} + e_{N+1|k} - A_K e_{N|k}}{S}}$.

Note that for every $\delta \leq (1 - \hat\lambda) / \hat\lambda$, equivalent to $1 - (1+\delta)  \hat\lambda \geq 0$, condition (\ref{eq:P-scaling}) holds for an $\alpha$ big enough. From this and $A_K^\top P A_K \preceq \hat\lambda P$, it follows that 
\begin{equation*}
S - (1+\delta)A_K^\top S A_K = \alpha \left( P - (1+ \delta) \cdot A_K^\top P A_K \right) \succeq \alpha \left( 1 - (1 + \delta) \cdot \hat\lambda \right)P \succeq Q + K^\top R K
\end{equation*}
which allows us to drop the terms in $x^\ast_{N|k}$ and get
\begin{equation*}
\begin{split}
\EE{\tilde{\mathcal{J}}_{k+1}} - \EE{\mathcal{J}^\ast_{k}} \leq -\norm{x^\ast_{0|k}}{Q} \hspace{-0.2cm} - \norm{\sat{u^\ast_{0|k}}}{R} \hspace{-0.2cm} + \left(1 + \frac{1}{\delta}\right) \EE{\norm{e_{N+1|k} - A_K e_{N|k}}{S}}.
\end{split}
\end{equation*}
From the parallelogram law, i.e. $\norm{a + b}{} \leq 2\norm{a}{} + 2 \norm{b}{}$, we obtain
\begin{equation*}
\begin{split}
\EE{\norm{e_{N+1|k} - A_K e_{N|k}}{S}} &\leq 2 \, \EE{\norm{e_{N+1|k}}{S}} + 2 \, \EE{\norm{e_{N|k}}{A_K^\top S A_K}} \\
&\leq 2 \left(\hat\lambda \,  \EE{\norm{e_{N|k}}{S}} + \Tr(SW) \right) + 2 \EE{\norm{e_{N|k}}{A_K^\top S A_K}} \\
&\leq 2 \left( (1+ \hat\lambda) \EE{\norm{e_{N|k}}{S}} + \Tr(SW)  \right) \leq 2 \left( \frac{1 + \hat\lambda}{1 - \hat\lambda} \Tr(SW) + \Tr(SW) \right) \\
& = \frac{4 }{1 - \hat\lambda} \Tr(SW).
\end{split}
\end{equation*}
Using the above result we obtain
\begin{equation*}
\begin{split}
\EE{\tilde{\mathcal{J}}_{k+1}} - \EE{\mathcal{J}^\ast_{k}}
&\leq -\norm{x^\ast_{0|k}}{Q} - \norm{\sat{u^\ast_{0|k}}}{R} + \left(1 + \frac{1}{\delta}\right) \frac{4 }{1 - \hat\lambda} \Tr(SW) \\
&=  -\ell\!\left(x^\ast_{0|k}, \varphi\left(u^\ast_{0|k}\right) \right) + \left(1 + \frac{1}{\delta}\right) \frac{4 }{1 - \hat\lambda} \Tr(SW).
\end{split}
\end{equation*}
Finally, by optimality we have $\EE{\mathcal{J}^\ast_{k+1}} \leq \EE{\tilde{\mathcal{J}}_{k+1}}$, hence the same inequality holds with $\mathcal{J}^\ast_{k+1}$ on the left-hand side. Taking expectations and summing from $k=0$ to $\bar N-1$ yields
\begin{equation*}
\EE{\mathcal{J}^\ast_{\bar N}} - \EE{\mathcal{J}^\ast_{0}}
\leq -\sum_{k=0}^{\bar N-1}\EE{\ell\!\left(x^\ast_{0|k}, \varphi\left(u^\ast_{0|k}\right) \right)} +  \bar{N} \left(1 + \frac{1}{\delta}\right) \frac{4 }{1 - \hat\lambda} \Tr(SW),
\end{equation*}
and since $\EE{\mathcal{J}^\ast_{\bar N}}\geq 0$, dividing by $\bar N$ and letting $\bar N \to \infty$ gives the result in \eqref{eq:nonlinear-cost-bound}. \end{proof}

\begin{remark}
     The choice of $S=\alpha P$ with $\alpha \geq 1$ is made only to obtain a obtain a clean bound proportional to $\Tr(SW)$. The argument remains valid for a generic terminal weight $S \succ 0$ satisfying a contraction condition of the form $S \succeq (1 + \delta) \cdot A_K^\top S A_K + Q + K^\top R K$, which would introduce scaling constants into the final result. We restrict ourselves to $S=\alpha P$ for clarity and compactness of the final bound.
\end{remark}

\begin{remark}
    A more interpretable, albeit mathematically moot, way of looking at the cost bound above may be obtained by close inspection of the term $\EE{\norm{e_{N+1|k} - A_K e_{N|k}}{S}}$. Under the terminal law $v_{N|k}=K z_{N|k}$ (hence $u_{N|k}=Kx_{N|k}$), one has
    \begin{equation*}
        \begin{split}
            \EE{\norm{e_{N+1|k} - A_K e_{N|k}}{S}} &= \EE{\norm{ Ae_{N|k} + B\left(\varphi\left( Kx_{N|k}\right) - Kz_{N|k} \right)  + w_{N+k} - Ae_{N|k} - B K e_{N|k}}{S}}, \\
            &= \EE{\norm{ B\left( \varphi\left(Kx_{N|k}\right) - Kx_{N|k} \right) + w_{N+k} }{S}}.
        \end{split}
    \end{equation*}
    Since $w_{N+k}$ is zero mean and independent of $x_{N|k}$ by Assumption \ref{ass:noise-iid},
        \begin{equation*}
        \begin{split}
            \EE{\norm{e_{N+1|k} - A_K e_{N|k}}{S}} &= \EE{\norm{ B\left( \varphi\left(Kx_{N|k}\right) - Kx_{N|k} \right)}{S}} + \EE{\norm{w_{N+k}}{S}}, \\
            &= \EE{\norm{ B\left( \varphi\left(Kx_{N|k}\right) - Kx_{N|k} \right)}{S}} + \Tr(SW).
        \end{split}
    \end{equation*}
    Hence, the cost bound can be interpreted as being driven by the saturation mismatch and the disturbance covariance, which aligns with our physical intuition and the established results on nonlinear SMPC stability \cite{McAllister_2023}. 
\end{remark}

\section*{References}
\bibliography{references}
\end{document}